\documentclass[11pt, fleqn]{article}

\title{On Clustering with Discounts}
\author{Shichuan Deng\thanks{IIIS, Tsinghua University, Email: \texttt{dsc15@mails.tsinghua.edu.cn}
}}
\date{}

\usepackage{amsmath,mathrsfs}
\usepackage{amsfonts,bbm}
\usepackage{amssymb}
\usepackage{graphicx,url}
\usepackage[margin=1in]{geometry}
\usepackage{mathrsfs}
\usepackage{float}
\usepackage[dvipsnames]{xcolor}
\usepackage[linesnumbered,ruled,vlined]{algorithm2e}

\usepackage{xspace}
\usepackage{amsthm}
\usepackage{enumitem}
\usepackage{hyperref}
\usepackage[capitalise]{cleveref}
\usepackage{tikz}
\usetikzlibrary{shapes,positioning,patterns}
\hypersetup{
    colorlinks=true,
    linkcolor=blue,
    urlcolor=cyan,
    citecolor=ForestGreen,
}

\DeclareMathOperator*{\argmin}{arg\,min}

\newcommand{\ceil}[1]{\left\lceil #1\right\rceil}
\newcommand{\etal}{\textit{et~al.}\xspace}
\newcommand{\etalcite}[1]{\textit{et~al.}~\cite{#1}}

\newcommand{\calI}{\mathcal{I}}

\newcommand{\calL}{\mathcal{L}}

\newcommand{\calC}{\mathcal{C}}
\newcommand{\calE}{\mathcal{E}}
\newcommand{\calF}{\mathcal{F}}
\newcommand{\calM}{\mathcal{M}}

\newcommand{\scr}[1]{\mathscr{#1}}

\newcommand{\E}{\mathbb{E}}

\newcommand{\baring}[1]{\overline{#1}}

\newcommand{\atau}{\alpha_\tau}
\newcommand{\btau}{\beta_\tau}

\newcommand{\opt}{\mathsf{OPT}}
\newcommand{\est}{\mathsf{EST}}

\newcommand{\fav}{f_{\mathsf{av}}}

\newcommand{\regret}{{\small\textsc{Re}}\xspace}
\newcommand{\optimal}{{\small\textsc{Opt}}\xspace}
\newcommand{\solution}{{\small\textsc{Sol}}\xspace}
\newcommand{\minregret}{{\small\textsc{Mr}}\xspace}
\newcommand{\minregretsol}{{\small\textsc{Mrs}}\xspace}

\newcommand{\defeq}{\stackrel{\text{def}}{=}}
\newcommand{\kmd}{{\small$k$\textsc{MedDis}}\xspace}
\newcommand{\matmd}{{\small\textsc{MatMedDis}}\xspace}
\newcommand{\knapmd}{{\small\textsc{KnapMedDis}}\xspace}
\newcommand{\unikmd}{{\small\textsc{Uni-}$k$\textsc{MedDis}}\xspace}

\newcommand{\unkcen}{{\small\textsc{USto}$k$\textsc{Cen}}\xspace}
\newcommand{\unmatcen}{{\small\textsc{UStoMatCen}}\xspace}
\newcommand{\unknapcen}{{\small\textsc{UStoKnapCen}}\xspace}

\newtheorem{theorem}{Theorem}
\newtheorem{lemma}[theorem]{Lemma}

\newtheorem{claim}[theorem]{Claim}

\theoremstyle{definition}
\newtheorem{definition}[theorem]{Definition}
\theoremstyle{remark}
\newtheorem*{remark*}{Remark}
\newtheorem*{note*}{Note}

\begin{document}

\maketitle

\begin{abstract}
    We study the $k$-median with discounts problem, wherein we are given clients with non-negative discounts and seek to open at most $k$ facilities. The goal is to minimize the sum of distances from each client to its nearest open facility which is discounted by its own discount value, with minimum contribution being zero. $k$-median with discounts unifies many classic clustering problems, e.g., $k$-center, $k$-median, $k$-facility $l$-centrum, etc.
    We obtain a bi-criteria constant-factor approximation using an iterative LP rounding algorithm. Our result improves the previously best approximation guarantee for $k$-median with discounts [Ganesh~\etal, ICALP'21].
    We also devise bi-criteria constant-factor approximation algorithms for the matroid and knapsack versions of median clustering with discounts.
\end{abstract}

\section{Introduction}

In this paper, we consider the $k$-median with discounts problem. Formally, in an instance of $k$-median with discounts (\kmd), given clients $\calC$, facilities $\calF$, a finite metric $\{c_{ij}\geq0:i,j\in\calC\cup\calF\}$ and discounts $\{r_j\geq0:j\in\calC\}$, our goal is to choose at most $k$ facilities $S\subseteq\calF$ such that the sum of discounted distances from each client to its nearest facility in $S$ is minimized, i.e., to minimize $\sum_{j\in\calC}(c_{jS}-r_j)^+$ where $c_{jS}=\min_{i\in S}c_{ij}$ and $a^+=\max\{a,0\}$. 
This problem recovers standard $k$-median~\cite{byrka2017improved,charikar2012dependent} by having $r_j=0$ for each $j\in\calC$. \kmd also recovers $k$-supplier~\cite{gonzalez1985clustering,hochbaum1985best} using the following reduction: By setting a uniform discount $r_j=R$ for all $j$, the instance has optimum 0 if and only if $R\geq R^\star$ where $R^\star$ is the optimum of the $k$-supplier problem on the same input. Since there are a polynomial number of possible values for $R^\star$, solving the $k$-supplier instance is equivalent to solving \kmd for each possible $R$. 

An algorithm is said to be a \emph{bi-criteria $(\alpha,\beta)$-approximation} for \kmd, if it always outputs a solution $\baring{S}$ s.t. $\sum_{j\in\calC}(c_{j\baring{S}}-\alpha\cdot r_j)^+\leq\beta\sum_{j\in\calC}(c_{jS^\star}-r_j)^+$, where $S^\star$ is the optimal solution.
Using the same reduction as above, we observe that such an algorithm recovers an $\alpha$-approximation to $k$-supplier and a $\beta$-approximation to $k$-median. Using the APX-hardness of $k$-supplier~\cite{hochbaum1986unified} and $k$-median~\cite{jain2002greedy}, one has $\alpha\geq3$ and $\beta\geq1+2/e$ unless $\mathrm{P=NP}$. Under current theoretical complexity assumptions, this rules out pure approximations and shows bi-criteria $(O(1),O(1))$-approximations are the strongest type of guarantee obtainable.
Guha and Munagala~\cite{guha2009exceeding} consider $k$-median with uniform discounts (\unikmd), where the input is the same as \kmd but the discounts are uniform. They directly use the primal-dual algorithm by Jain and Vazirani~\cite{jain2001approximation} and give a bi-criteria $(9,6)$-approximation algorithm. They develop a constant factor approximation for unassigned stochastic $k$-center using \unikmd.
Very recently, Ganesh~\etalcite{ganesh2021universal} give a bi-criteria $(9,6)$-approximation algorithm for \kmd, using a similar primal-dual algorithm. They also use this algorithm as a key ingredient in their bi-criteria constant-factor approximations for universal $k$-clustering problems.

\unikmd is also closely related to \emph{$k$-facility $l$-centrum}, where the input is the same as $k$-median, but the objective is the sum of $l$ largest connection costs. This problem unifies $k$-supplier and $k$-median via $l\in\{1,|\calC|\}$. Given a bi-criteria $(\alpha,\beta)$-approximation for \unikmd, one recovers a $\max\{\alpha,\beta\}$-approximation for $k$-facility $l$-centrum, by exhaustive search for the $l$-th largest connection cost $R$ in the optimum and setting a uniform discount $R$ \cite{chakrabarty2019approximation}. An improved bi-criteria $(5,5+\epsilon)$-approximation for \unikmd is given in \cite{chakrabarty2019approximation} in the more general setting of multi-dimensional discounts and ordered $k$-median.

\paragraph{Our contributions}
We present an LP rounding algorithm for \kmd and achieve better approximation factors than the previously-known bi-criteria $(9,6)$-approximation \cite{ganesh2021universal}. We employ a modestly improved iterative rounding framework by Krishnaswamy~\etalcite{krishnaswamy2018constant}, inspired by the ``quarter ball chasing'' technique in \cite{gupta2021structual}. 

In the original iterative rounding framework \cite{krishnaswamy2018constant}, one starts with a solution to the natural relaxation, and obtains an (almost) integral solution by iteratively modifying an auxiliary LP. Like several previous LP based algorithms for median clustering \cite{charikar2012dependent,swamy2016improved}, the main idea is to have a basic solution such that the tight constraints at this solution consists of laminar families. In \cite{krishnaswamy2018constant}, the laminar family is obtained by dynamically maintaining a ``core set'' of clients and creating packing constraints indexed by these clients. Gupta~\etalcite{gupta2021structual} further refines the framework by adding an additional laminar family in the auxiliary LP and analyzing the structure of the basic solutions.

To obtain an approximation for \kmd, we adapt the framework by having a new objective, where the contribution of assigning facility $i$ to client $j$ is $(\hat c_{ij}-\tau r_j)^+$ for $\tau>1$ and discretized cost function $\hat c\geq c$. Compared with the real contribution, this objective prescribes a larger discount for each client, which in turn provides an upper bound on the auxiliary objective after we discretize the given metric $c$. Since our formulation is free of outliers, we are able to have a simpler analysis of the approximation guarantee for using two laminar families in the auxiliary LP.

The LP based algorithm also gives rise to the first constant factor approximations for matroid median with discounts (\matmd) and knapsack median with discounts (\knapmd). In these two problems, the input is the same as \kmd, except that the cardinality constraint $k$ is absent. We need to either choose an independent set of facilities in a given matroid $\calM=(\calF,\calI)$, or choose weighted facilities that have a total weight no more than a given threshold. The reduction from matroid center/median, knapsack center/median to \matmd, \knapmd follows similarly.

We also obtain constant-factor approximations for clustering problems with uncertain and stochastic clients. These are application results which rely on previously-existing frameworks, e.g., the $O(1)$-approximation algorithm for unassigned stochastic $k$-center by Guha and Munagala \cite{guha2009exceeding} and the bi-criteria $(O(1),O(1))$-approximation for universal $k$-median by Ganesh~\etalcite{ganesh2021universal}. See~\cref{section:application} for more details.


\paragraph{Related work}
$k$-center and $k$-median are two of the most fundamental clustering problems. $k$-center is NP-hard to approximate to a factor better than $2$ \cite{hsu1979bottleneck}, and simple 2-approximations  are developed \cite{gonzalez1985clustering,hochbaum1985best}. There is a long line of research for $k$-median, with approximations given by primal-dual methods~\cite{jain2002greedy,jain2001approximation,li2016approximating}, local search~\cite{arya2004local} and LP rounding~\cite{charikar2002constant,charikar2012dependent}. Currently the best approximation ratio is $2.675+\epsilon$ due to Byrka~\etalcite{byrka2017improved}. No performance guarantee better than $1+2/e$ is obtainable in polynomial time unless $\mathrm{P}=\mathrm{NP}$~\cite{jain2002greedy}.

Clustering problems with stronger combinatorial constraints are extensively studied. Hochbaum and Shmoys~\cite{hochbaum1986unified} give a tight 3-approximation for knapsack center. Chen~\etalcite{chen2016matroid} give a tight 3-approximation for matroid center. Several approximations are developed for median clustering~\cite{charikar2012dependent,krishnaswamy2011matroid,kumar2012constant,swamy2016improved} with the current best ratios $7.081$ \cite{krishnaswamy2018constant} for matroid median and $6.387+\epsilon$ \cite{gupta2021structual} for knapsack median. Many variants are considered in the literature, e.g., robust matroid and knapsack center~\cite{chen2016matroid,harris2019lottery} where a certain number of clients can be discarded, and the more general robust $\mathscr{F}$-center problem~\cite{chakrabarty2018generalized}, where some clients are discarded and $\mathscr{F}$ is an independence system containing all feasible solutions.

Cormode and McGregor~\cite{cormode2008approximation} introduce the study of stochastic $k$-center and obtain a bi-criteria constant-factor approximation. Guha and Munagala~\cite{guha2009exceeding} improve their result and obtain $O(1)$-approximations for both the unassigned and assigned versions of stochastic $k$-center. Huang and Li~\cite{huang2017stochastic} consider the unassigned version of stochastic $k$-center and devise the first PTAS for fixed $k$ and constant-dimensional Euclidean metrics. Alipour and Jafari~\cite{alipour2018improvement} study multiple variants of assigned stochastic $k$-center and give various constant-factor approximations, some of which are later improved by Alipour~\cite{alipour2021improvements}.


\subsection{Organization}

In~\cref{section:k:discount}, we present a bi-criteria $(O(1),O(1))$-approximation for \kmd.
In~\cref{section:matroid:discount}, we present a bi-criteria $(O(1),O(1))$-approximation for \matmd.
In~\cref{section:knapsack:discount}, we present a bi-criteria $(O(1),O(1))$-approximation for \knapmd.
We then provide application in~\cref{section:application}, and defer some proof details to the appendix.

\section{$k$-Median with Discounts}
\label{section:k:discount}
\subsection{The Natural Relaxation}
In this section, we consider \kmd. The natural relaxation is given as follows, where $x_{ij}\in[0,1]$ is the extent we assign facility $i$ to client $j$, and $y_i\in[0,1]$ is the extent we open facility $i$. 
\begin{alignat*}{3}
    \text{min}\quad&&\sum_{j\in\calC}\sum_{i\in\calF}x_{ij}&(c_{ij}-r_j)^+\quad\tag{$\mathrm{LP\text{-}}k$}\label{lp:k:discount}\\
    \text{s.t.}\quad
    &&\sum_{i\in\calF}x_{ij}&=1&&\forall j\in\calC\\
    &&\sum_{i\in\calF}y_i&\leq k\\
    &&0\leq x_{ij}&\leq y_i\leq 1&&\forall i\in\calF,j\in\calC.
\end{alignat*}

Let $\opt_1\geq0$ be the optimum of  the instance of \kmd.
The optimum of \ref{lp:k:discount} is at most $\opt_{1}$, since the integral solution induced by the optimal solution is feasible, and the objective is equal to $\opt_1$. Fix an optimal solution $(\bar x,\bar y)$ to \ref{lp:k:discount} in what follows. We assume w.l.o.g. that $(\bar x,\bar y)$ is \emph{distance-optimal}, i.e., whenever there exists $\bar x_{ij}>0$, for each $i'$ s.t. $c_{i'j}<c_{ij}$, we have $\bar x_{i'j}=\bar y_{i'}$. We note that though this is NOT guaranteed by the LP objective due to the discounts, we can always modify the solution such that it becomes distance-optimal while keeping the objective value unchanged.

\subsection{Metric Discretization and Auxiliary LP}

In the sequel, we adapt the iterative rounding framework in \cite{krishnaswamy2018constant}.
W.l.o.g., we assume $c_{pq}\geq1$ for any $p,q\in\calF\cup\calC$ that are \emph{not} co-located. We fix $\tau>1$ and discretize the metric $c$ as follows: Let $D_{-2}=-1$, $D_{-1}=0$, and $D_\ell=a\tau^\ell$ for any $\ell\in\mathbb{Z}_{\geq0}$, where $a=\tau^b$ and $b\in [0,1)$ is a parameter we will determine later. For any $p,q\in\calF\cup\calC$, define $\hat c_{pq}=\min_{\ell:D_\ell\geq c_{pq}}D_\ell$. One notices that $\hat c$ is not necessarily a metric.
We obtain the following lemma.

\begin{lemma}\label{lemma:metric:expectation}
If $b$ is a random variable uniformly distributed on $[0,1)$, for any $i\in\calF$, $j\in\calC$, one has $\E_b[(\hat c_{ij}-\tau r_j)^+]\leq\frac{\tau-1}{\ln\tau}(c_{ij}-r_j)^+$.
\end{lemma}

\begin{proof}
When $c_{ij}=0$, the expectation is 0 and the inequality is trivial. We assume $c_{ij}\geq1$ in what follows.
Let $c_{ij}=\tau^{s+u}$ where $s\in\mathbb{Z}$ and $u\in[0,1)$. Since $D_\ell=a\tau^{\ell}=\tau^{\ell+b}$ for $\ell\geq0$ and $b$ is distributed uniformly on $[0,1)$, $\hat c_{ij}=\tau^{s+1+b}$ when $b<u$, and $\hat c_{ij}=\tau^{s+b}$ when $b\geq u$.
This shows $\hat c_{ij}/c_{ij}\in[1,\tau)$, hence
\begin{align*}
\E_b[(\hat c_{ij}-\tau r_j)^+]&\leq\E_b\left[\left(\hat c_{ij}-\frac{\hat c_{ij}}{c_{ij}}r_j\right)^+\right]
=(c_{ij}-r_j)^+\E_b\left[\hat c_{ij}/c_{ij}\right]\\
&=(c_{ij}-r_j)^+\Big(\int_0^u\tau^{1+b-u}+\int_u^1\tau^{b-u}\Big)=\frac{\tau-1}{\ln\tau}(c_{ij}-r_j)^+.\tag*{\qedhere}
\end{align*}
\end{proof}

Using the obtained solution $(\bar x,\bar y)$ to \ref{lp:k:discount}, we employ the standard facility duplication technique to make sure that $\bar x_{ij}\in\{0,\bar y_i\}$ (see, e.g., \cite{charikar2012dependent,krishnaswamy2018constant}). Let $F_j=\{i\in\calF:\bar x_{ij}>0\}$ be the \emph{outer ball} of $j\in\calC$ and thus $\bar y(F_j)=\sum_{i\in F_j}\bar y_i=1$. Let $\ell_j\in\mathbb{Z}$ be the smallest integer such that $\hat c_{ij}\leq D_{\ell_j}$ for each $i\in F_j$ called the \emph{radius level} of $j$, and $B_j=\{i\in F_j:\hat c_{ij}\leq D_{\ell_j-1}\}$ be the \emph{inner ball} of $j$. Let $C_0\leftarrow \calC$, $C_1\leftarrow\emptyset$ and $C^\star\leftarrow\emptyset$ initially, and we will iteratively move \emph{all} of $C_0$ to $C_1$, and maintain a subset $C^\star\subseteq C_1$. We first define the auxiliary LP.
\begin{alignat}{3}
    \text{min}\quad&
        \sum_{j\in C_1}\Bigg(\sum_{i\in B_j}y_i(\hat c_{ij}-\tau r_j)^++(1-y(B_j))( D_{\ell_j}-\tau r_j)^+\Bigg)\notag\\
        &+\sum_{j\in C_0}\sum_{i\in F_j}y_i(\hat c_{ij}-\tau r_j)^+\tag{$\mathrm{IR\text{-}}k$}\label{lp:iterative:k}\\
    \text{s.t.}\quad& y(F_j)=1  \quad \forall j\in C_0
    \quad\quad\quad y(\calF)\leq k \notag\\
    & y(B_j)\leq1  \quad \forall j\in C_1
    \quad\quad\quad y_i\in[0,1] \quad \forall i\in\calF.\notag\\
    & y(F_j)=1  \quad \forall j\in C^\star\notag
\end{alignat}
\begin{lemma}\label{lemma:iterative:zero}
$\bar y$ is feasible to \ref{lp:iterative:k}. There exists $b\in[0,1)$ such that the objective of $\bar y$ to \ref{lp:iterative:k} is at most $\frac{\tau-1}{\ln\tau}\opt_1$.
\end{lemma}
\begin{proof}
In the beginning, we have $C_0=\calC$ and the other two are empty, thus $\bar y$ satisfies all the constraints by definition of $F_j$.
We start by letting $b$ be uniformly distributed on $[0,1)$. Using \cref{lemma:metric:expectation} and the linearity of expectation, the expectation of \ref{lp:iterative:k} under $\bar y$ is
\begin{equation*}
\E_b\left[\sum_{j\in\calC}\sum_{i\in\calF}\bar x_{ij}(\hat c_{ij}-\tau r_j)^+\right]
\leq\frac{\tau-1}{\ln\tau}\sum_{i,j}\bar x_{ij}(c_{ij}-r_j)^+.
\end{equation*}
Since the last sum above is the objective of \ref{lp:k:discount}, the expectation is at most $\frac{\tau-1}{\ln\tau}\opt_1$.

If we increase $b$ continuously from 0 to $1^-$, $(\hat c_{ij}-\tau r_j)^+$ is piece-wise non-decreasing and right-continuous for any $(i,j)$ by definition of $\hat c$. Specifically, suppose $c_{ij}=\tau^{s+u}$ for $s\in\mathbb{Z}$ and $u\in[0,1)$, then the expression is non-decreasing on $[0,u)$ and $[u,1)$. Therefore, the objective of \ref{lp:iterative:k} under $\bar y$, is also piece-wise non-decreasing and right-continuous, with respect to $b$ and a polynomial number of intervals. Therefore, it is easy to find $b\in[0,1)$ that minimizes the objective, which is at most the expectation.
\end{proof}

\subsection{Analysis}
Suppose we choose such $b\in[0,1)$ using \cref{lemma:iterative:zero} in the sequel. We use the iterative rounding algorithm to iteratively change \ref{lp:iterative:k}, maintain a feasible solution, and keep the objective value non-increasing. The rigorous procedures are in \cref{algo:iterative}, and here we set the \emph{step size} $h=2$.

\begin{algorithm}[ht]
\caption{{\small\textsc{IterRound}} \cite{gupta2021structual,krishnaswamy2018constant}}\label{algo:iterative}
\SetKwInOut{Input}{Input}\SetKwInOut{Output}{Output}
\SetKwBlock{Update}{{update-}$C^\star (j,h)$}{end}
\DontPrintSemicolon
\Input{outer balls $\{F_j:j\in\calC\}$, radius levels $\{\ell_j:j\in\calC\}$, inner balls $\{B_j:j\in\calC\}$, step size $h\in\{1,2\}$}
\Output{an integral solution $y^\star$ to \ref{lp:iterative:k}}
$C_0\leftarrow\calC,C_1\leftarrow\emptyset,C^\star \leftarrow \emptyset$\;
\While{true}{
find an optimal basic solution $y^\star$ to~\ref{lp:iterative:k}\;
\uIf{$\exists j\in C_0$, i.e., $C_0\neq\emptyset$}{
$C_0\leftarrow C_0\setminus\{j\},C_1\leftarrow C_1\cup\{j\},B_j\leftarrow\{i\in F_j:\hat c_{ij}\leq D_{\ell_j-1}\}$, \textbf{update-}$C^\star (j,h)$\;
}
\uElseIf{$\exists j\in C_1$ s.t. $y^\star(B_j)=1$}{
$\ell_j\leftarrow \ell_j-1,F_j\leftarrow B_j,B_j\leftarrow\{i\in F_j:\hat c_{ij}\leq D_{\ell_j-1}\}$, \textbf{update-}$C^\star (j,h)$\;
}
\lElse{
break
}}
\Return $y^\star$\;
\Update{
\lIf{$\nexists j'\in C^\star$ with $\ell_{j'}\leq \ell_j$ and $F_j\cap F_{j'}\neq\emptyset$}{$C^\star\leftarrow C^\star  \cup\{j\}$,
remove from $C^\star $ all $j'$ such that $F_j\cap F_{j'}\neq \emptyset$ and $\ell_{j'}\geq\ell_j+h$
}}
\end{algorithm}

\begin{lemma}\label{lemma:iterative:one}
\cref{algo:iterative} returns an integral $y^\star$ in polynomial time, i.e., $y^\star\in\{0,1\}^\calF$, and $C_0=\emptyset$, $C_1=\calC$.
\begin{proof}
We first notice that, for each $j\in\calC$, it can only enter $C_1$ from $C_0$ once. In each iteration, we either move some $j$ from $C_0$ to $C_1$, or reduce the radius level $\ell_j$ of some $j\in C_1$ by 1. Since the latter happens only when $y^\star(B_j)=1$, and by definition of $B_j$, we must have $\ell_j\geq0$, otherwise $B_j=\{i\in F_j:\hat c_{ij}\leq D_{-2}=-1\}=\emptyset$, contradicting $y^\star(B_j)=1$. There are $O(\log_\tau \Delta)$ possible radius levels, where $\Delta=\max_{i\in\calF,j\in\calC}c_{ij}$, thus the algorithm returns in polynomial time.

When \cref{algo:iterative} returns, none of the constraints in \ref{lp:iterative:k} corresponding to $j\in C_0$ or $j\in C_1$ is tight. By definition of the subroutine $\mathbf{update}\textbf{-}C^\star$, for each $j\neq j'$ in $C^\star$, we have $F_j\cap F_{j'}=\emptyset$ unless $|\ell_j-\ell_{j'}|=1$ (recall that $h=2$), thus the remaining tight constraints at $y^\star$ in \ref{lp:iterative:k} form two laminar families, i.e., those in $C^\star$ with \emph{odd} radius levels and those with \emph{even} radius levels. Hence $y^\star$ is an integral basic solution.
\end{proof}

\end{lemma}
\begin{lemma}\label{lemma:iterative:two}
After each iteration of \cref{algo:iterative}, $y^\star$ is feasible to the new LP with a no larger objective.
\end{lemma}
\begin{proof}
We consider two cases. In the first case, there exists some $j\in C_0$. The old contribution of $j$ to \ref{lp:iterative:k} is $\sum_{i\in F_j}y^\star_i(\hat c_{ij}-\tau r_j)^+$. After we move $j$ to $C_1$, because $y^\star(B_j)\leq y^\star(F_j)=1$, $y^\star$ satisfies the new constraints corresponding to $j$ in $C_1$ and $C^\star$ (if it is added to $C^\star$), and its contribution becomes $\sum_{i\in B_j}y^\star_i( \hat c_{ij}-\tau r_j)^++(1-y^\star(B_j))( D_{\ell_j}-\tau r_j)^+$. Since $F_j$ is partitioned into $B_j\cup(F_j\setminus B_j)$ and each $i\in F_j\setminus B_j$ must satisfy $\hat c_{ij}=D_{\ell_j}$ by definition, the contribution of $j$ stays the same.

In the second case, there exists $y^\star(B_j)=1$ for some $j\in C_1$. The old contribution of $j$ to \ref{lp:iterative:k} is $\sum_{i\in B_j}y^\star_i(\hat c_{ij}-\tau r_j)^+$ since $y^\star(B_j)=1$. To avoid confusion, let $F_j',B_j',\ell_j'$ denote the modified values of $j$ after the iteration. After we reduce the radius level of $j$ and attempt to add $j$ to $C^\star$, the new $F_j'=B_j$ satisfies $y^\star(F_j')=1$, hence $y^\star$ still satisfies the constraints, if $j$ is indeed added to $C^\star$. The contribution of $j$ becomes $\sum_{i\in B_j'}y^\star_i(\hat c_{ij}-\tau r_j)^++(1-y^\star(B_j'))(D_{\ell_j'}-\tau r_j)^+$, where $\ell_j'=\ell_j-1$. Since the new $F_j'=B_j$ is partitioned into $B_j'\cup(F_j'\setminus B_j')$, and each $i\in F_j'\setminus B_j'$ must satisfy $\hat c_{ij}=D_{\ell_j'}$ by definition, the contribution is unchanged.
\end{proof}

\begin{lemma}\label{lemma:iterative:three}
After \cref{algo:iterative} returns $y^\star$, its objective in \ref{lp:iterative:k} is at most that of $\bar y$ in the original auxiliary relaxation. Moreover, for each $j\in \calC$, there exists $i\in\calF$ such that $y^\star_i=1$ and $c_{ij}\leq\frac{3\tau^h-1}{\tau^h-1}D_{\ell_j}$, $h\in\{1,2\}$.
\end{lemma}
\begin{proof}
The first assertion follows from \cref{lemma:iterative:two}, since the objective of $y^\star$ is non-increasing within each iteration, and we find an optimal basic solution at the beginning of each iteration. For the second assertion, we first notice that $y^\star$ is an integral solution in $\{0,1\}^\calF$. 

Due to $D_{-2}=-1$, every client starts with a radius level of at least $-1$. When $\ell_j=-1$ during \cref{algo:iterative}, its inner ball becomes $B_j=\{i\in F_j:\hat c_{ij}\leq D_{-2}=-1\}=\emptyset$, thus $y^\star(B_j)=0$ and $\ell_j$ stays $-1$ thereafter. In other words, $\ell_j\geq-1$ for each $j$ throughout \cref{algo:iterative}.

If $\ell_j=-1$ and $j\in C^\star$ during \cref{algo:iterative}, according to the subroutine $\mathbf{update}\textbf{-}C^\star$, only $j'$ with $\ell_{j'}<\ell_j$ can remove $j$ from $C^\star$, but this implies $\ell_{j'}\leq-2$ which is impossible. Therefore, whenever there exists $j\in C^\star$ with $\ell_j=-1$, it cannot be removed from $C^\star$. We first need the following claim.
\begin{claim}\label{claim:near:facility}
If $j$ is added to $C^\star$ when it has radius level $\ell$, there exists $i\in\calF$ with $y^\star_i=1$ in the final solution, and $c_{ij}\leq\frac{\tau^h+1}{\tau^h-1}D_\ell$, $h\in\{1,2\}$.
\end{claim}
\begin{proof}
We prove it using induction on $\ell$. If $\ell=-1$, the claim follows from the argument above since $j$ cannot be removed from $C^\star$, and each $i\in F_j$ satisfies $c_{ij}\leq D_\ell$. 
Suppose the claim holds for radius levels up to $\ell-1$, we consider $j$ added to $C^\star$ with radius level $\ell$. If it remains in $C^\star$ to the end, then the claim follows similarly as $c_{ij}\leq D_\ell$ for each $i\in F_j$. If it is removed by another $j'$ later in the algorithm, it means that $\ell_{j'}\leq\ell_j-h\leq\ell-h$ since the radius level of $j$ is non-increasing. Using the induction hypothesis, there exists $i\in\calF$ with $y^\star_i=1$ and $c_{ij'}\leq \frac{\tau^h+1}{\tau^h-1}D_{\ell_j'}\leq \frac{\tau^h+1}{\tau^h-1}D_{\ell}/\tau^h$, thus using the triangle inequality, we have $c_{ij}\leq D_{\ell_{j'}}+D_{\ell_j}+\frac{\tau^h+1}{\tau^h-1}D_{\ell}/\tau^h\leq D_\ell(1+\frac{1}{\tau^h}+\frac{\tau^h+1}{\tau^h(\tau^h-1)})=\frac{\tau^h+1}{\tau^h-1}D_\ell$. This finishes the induction.
\end{proof}

For any client $j$ with a \emph{final} radius level of $\ell$, when we reduce $\ell_j$ to $\ell$ and invoke the subroutine on it, if indeed we can add $j$ to $C^\star$, we directly invoke \cref{claim:near:facility} and get the desired result.
Otherwise, $j$ cannot be added to $C^\star$ due to there being $j'\in C^\star$ s.t. $F_j\cap F_{j'}\neq\emptyset$ and $\ell_{j'}\leq\ell_j=\ell$, in which case we use \cref{claim:near:facility} \emph{on the time we add $j'$ to $C^\star$ with radius level $\ell_{j'}$} and conclude the existence of $i\in\calF$ with $y^\star_i=1$ and $c_{ij}\leq c_{ij'}+D_{\ell_{j'}}+D_{\ell_j}\leq \frac{\tau^h+1}{\tau^h-1}D_{\ell_{j'}}+2D_{\ell_j}\leq\frac{3\tau^h-1}{\tau^h-1}D_{\ell}$.
\end{proof}

Using the final output $y^\star$, we define the solution $\bar F=\{i\in\calF:y^\star_i=1\}$ and remove any co-located copies of facilities, if there are any (notice this does not affect \cref{lemma:iterative:three}). We prove the following main theorem.
\begin{theorem}\label{theorem:main:iterative:k}
Let $\alpha_\tau=\frac{\tau(3\tau^2-1)}{\tau^2-1}$ and $\beta_\tau=\frac{3\tau^2-1}{(\tau+1)\ln\tau}$. There exists a polynomial time bi-criteria $(\alpha_\tau,\beta_\tau)$-approximation algorithm for \kmd for any $\tau>1$.
\end{theorem}
\begin{proof}
We consider each client $j$ and its distance to $\bar F$, i.e., $c_{j\bar F}$. Since $y^\star$ is an integral solution, $C_0=\emptyset$ and $C_1=\calC$ when \cref{algo:iterative} finishes, it follows that $y^\star(B_j)=0$ for each $j\in\calC$ by definition of the algorithm. Using \cref{lemma:iterative:two}, the objective of $y^\star$ is at most that of $\bar y$ \emph{before we run the algorithm}, that is,
\begin{equation}
    \sum_{j\in\calC}(D_{\ell_j}-\tau r_j)^+
    \leq\sum_{j\in\calC}\sum_{i\in F_j}\bar y_i(\hat c_{ij}-\tau r_j)^+.\label{eqn:objective:decay}
\end{equation}
Here, we omit the sum over each $i\in B_j$ since each such $y^\star_i=0$ follows from $y^\star(B_j)=0$.
For each $j$, we have $c_{j\bar F}\leq\frac{3\tau^2-1}{\tau^2-1}D_{\ell_j}$ using \cref{lemma:iterative:three}, because we run \cref{algo:iterative} with $h=2$. Summing over all clients, one has
\begin{align}
    \sum_{j\in\calC}\left(c_{j\bar F}-\frac{\tau(3\tau^2-1)}{\tau^2-1}r_j\right)^+
    &\leq\sum_{j\in\calC} \left(\frac{3\tau^2-1}{\tau^2-1}D_{\ell_j}-\frac{\tau(3\tau^2-1)}{\tau^2-1}r_j\right)^+\notag\\
    &=\frac{3\tau^2-1}{\tau^2-1}\sum_{j\in\calC}\left(D_{\ell_j}-\tau r_j\right)^+,
    \label{eqn:objective:decay2}
\end{align}

Combining \eqref{eqn:objective:decay} and \eqref{eqn:objective:decay2} and using \cref{lemma:iterative:zero} for the initial objective of $\bar y$, we obtain
\[
\sum_{j\in\calC}\left(c_{j\bar F}-\frac{\tau(3\tau^2-1)}{\tau^2-1}r_j\right)^+\leq\frac{3\tau^2-1}{(\tau+1)\ln\tau}\opt_1,
\]
whence the theorem follows.
\end{proof}

\begin{remark*}
\cref{theorem:main:iterative:k} provides a trade-off between the approximation factors. For instance, one can choose $\tau=1.91$ s.t. $(\alpha_\tau,\beta_\tau)<(7.173,5.281)$ minimizing $\beta_\tau$, or $\tau=1.592$ s.t. $(\alpha_\tau,\beta_\tau)<(6.851,5.479)$ minimizing $\alpha_\tau$. Both improve previous bi-criteria $(9,6)$-approximation algorithms for \kmd \cite{ganesh2021universal,guha2009exceeding}.
\end{remark*}

\section{Matroid Median with Discounts}\label{section:matroid:discount}

In this section, we consider \matmd. Formally, \matmd has the same input as \kmd except that we need to open facilities that constitute an independent set of an input matroid $\calM=(\calF,\calI)$, instead of having an upper bound on the number of open facilities. \kmd is a special case of \matmd where the given matroid is uniform with rank $k$, that is, $\calI=\{F\subseteq\calF:|F|\leq k\}$.
The natural relaxation is the same as \ref{lp:k:discount}, except that we replace the cardinality constraint $y(\calF)\leq k$ with the constraints of a matroid polytope, that is, for $r_\calM:2^\calF\rightarrow\mathbb{Z}_{\geq0}$ as the rank function of $\calM$, $y(S)\leq r_\calM(S)$ for each $S\subseteq\calF$, following a classic result by Edmonds~\cite{edmonds2001submodular}.

Our algorithm proceeds very similarly to \kmd. We solve the relaxation and obtain a fractional solution $(\bar x,\bar y)$ that is distance-optimal. We discretize the metric into $\hat c$ and construct the same auxiliary LP as \ref{lp:iterative:k}, except that we replace the cardinality constraint with matroid constraints, too. The discretized metric is constructed in a way such that its objective is at most $\frac{\tau-1}{\ln\tau}$ times the optimum, akin to \cref{lemma:metric:expectation}.

We use \cref{algo:iterative} with a smaller step size $h=1$. This is because, unlike \kmd with the cardinality constraint, the matroid constraints in the auxiliary LP are non-trivial and the algorithm only admits an integral solution if the sets $\{F_j:j\in C^\star\}$ is a \emph{single} laminar family, as in \cite{krishnaswamy2018constant}.
We obtain an integral solution and define the solution $\bar F$ in the same way. Using the same arguments as \cref{theorem:main:iterative:k}, we obtain the following result. When $\tau=2.36$, we have $(\alpha_\tau',\beta_\tau')<(10.551,7.081)$, recovering the current best result for matroid median \cite{krishnaswamy2018constant}.
\begin{theorem}\label{theorem:main:iterative:matroid}
Let $\alpha_\tau'=\frac{\tau(3\tau-1)}{\tau-1}$ and $\beta_\tau'=\frac{3\tau-1}{\ln\tau}$. There exists a polynomial time bi-criteria $(\alpha_\tau',\beta_\tau')$-approximation algorithm for \matmd for any $\tau>1$.
\end{theorem}
\section{Knapsack Median with Discounts}\label{section:knapsack:discount}

In \knapmd, we are given a knapsack constraint instead of a cardinality constraint as in \kmd. Formally, each facility $i\in\calF$ has a weight $w_i\geq0$ and we need to open facilities $F\subseteq\calF$ with a combined weight no more than a given threshold $W$, that is, $w(F)=\sum_{i\in F}w_i\leq W$. \kmd is a special case of \knapmd where each $w_i=1$ and $W=k$.

\subsection{Sparsify the Instance}
The natural relaxation for the standard knapsack median problem has an unbounded integrality gap~\cite{kumar2012constant}. We overcome it by adapting the pre-processing by Krishnaswamy~\etalcite{krishnaswamy2018constant}. First, we let $n=|\calC|$, $n_0=|\calF\cup\calC|$ and fix an unknown optimal solution $F^\star$ to the problem. Let $\opt_2\geq0$ be the optimal objective thereof, that is, $\opt_2=\sum_{j\in\calC}(c_{jF^\star}-r_j)^+$, and $c_0\geq0$ be the largest contribution of any single client to it, that is, $c_0=\max_{j\in\calC}(c_{jF^\star}-r_j)^+$. It follows that $\opt_2\in[c_0,nc_0]$.

Fix a small constant $\epsilon>0$.
There are a polynomial number of possible values for $c_0$, and for each $c_0>0$, there are $O(\log_{1+\epsilon}n)$ possible values for $\est=(1+\epsilon)^{\ceil{\log_{1+\epsilon}\opt_2}}$. Thus, we enumerate all possible pairs $(c_0,\est)$ and assume it is as desired in what follows, i.e., $c_0$ is exactly the largest contribution of any client in the unknown optimum, and either $\est=0$ when $c_0=0$ or $\est=(1+\epsilon)^{\ceil{\log_{1+\epsilon}\opt_2}}$ when $c_0>0$.

For a \knapmd instance $\scr{I}=(\calF,\calC,c,r_j,w_i,W)$, we start by creating a new ``sparse'' instance where some facility set $F_0$ is pre-selected and some clients are pre-connected to $F_0$. The formal definition is given below.

\begin{definition}\label{definition:sparse:knapsack}
Let $\scr{J}=(\calF,\calC'\subseteq\calC,c,r_j,w_i,W,F_0)$ be an \emph{extended} instance of \knapmd, $\rho,\delta\in(0,1)$, and $F^\star$ be a solution to $\scr{J}$ with objective at most $\est\geq0$, that is, $w(F^\star)\leq W$ and $\sum_{j\in\calC'}(c_{jF^\star}-r_j)^+\leq\est$. Let $\kappa^\star_p=\argmin_{i\in F^\star}c_{pi}$ for any $p\in\calF\cup\calC'$, where ties are broken arbitrarily and consistently.
We say that $\scr{J}$ is $(\rho,\delta,\est)$-sparse w.r.t. $F^\star$ if
\begin{enumerate}[label=(\ref{definition:sparse:knapsack}.\arabic*), ref=(\ref{definition:sparse:knapsack}.\arabic*), leftmargin=1.2cm]
    \item\label{item:sparse1:knapsack} $\forall i\in F^\star\setminus F_0$, $\sum_{j\in\calC':\kappa^\star_j=i}(c_{ij}-r_j)^+\leq\rho\est$,
    \item\label{item:sparse2:knapsack} $\forall p\in\calF\cup\calC'$, $\sum\limits_{j\in\calC':c_{jp}\leq\delta c_{pF^\star}}(c_{pF^\star}-r_j/(1-\delta))^+\leq\rho\est$.
\end{enumerate}
\end{definition}

\begin{theorem}\label{theorem:sparsifying:knapsack}
For $\rho,\delta\in(0,1)$, $\scr{I}=(\calF,\calC,c,r_j,w_i,W)$ an instance of \knapmd, and an upper bound $\est$ on the objective $\opt_2$ of the optimal solution $F^\star$, there exists an algorithm that outputs $n_0^{O(1/(\rho-\rho\delta))}$ many extended instances, such that there exists one output instance $\scr{J}=(\calF,\calC',c,r_j,w_i,W,F_0)$ that satisfies
\begin{enumerate}[label=(\ref{theorem:sparsifying:knapsack}.\arabic*), ref=(\ref{theorem:sparsifying:knapsack}.\arabic*), leftmargin=1.2cm]
    \item\label{item:sparse3:knapsack} $\scr{J}$ is $(\rho,\delta,\est)$-sparse w.r.t. $F^\star$ and $F_0\subseteq F^\star$,
    \item\label{item:sparse4:knapsack} $\frac{1-\delta}{1+\delta}\sum\limits_{j\in\calC\setminus\calC'}\left(c_{jF_0}-\frac{1+\delta}{1-\delta}r_j\right)^++\sum\limits_{j\in\calC'}(c_{jF^\star}-r_j)^+\leq\est$.
\end{enumerate}
\end{theorem}
\begin{proof}
We begin with the premise that $F^\star$ is known to us and remedy this requirement later. We construct one such extended instance in the following two phases. Set $F_0\leftarrow\emptyset$ and $\calC'\leftarrow\calC$ initially. It follows from the conditions in the theorem that $\est\geq\sum_{j\in\calC}(c_{jF^\star}-r_j)^+$.

First, iteratively for each $i\in F^\star\setminus F_0$ that satisfies $\sum_{j\in\calC':\kappa^\star_j=i}(c_{ij}-r_j)^+>\rho\est$, we set $F_0\leftarrow F_0\cup\{i\}$. After this phase, for each $i\in F^\star\setminus F_0$, \ref{item:sparse1:knapsack} is satisfied. Because $\kappa^\star_j$ is the nearest facility to $j$ in $F^\star$, $\est\geq\sum_{i\in F^\star}\sum_{j:\kappa^\star_j=i}(c_{ij}-r_j)^+$ and at most $O(1/\rho)$ facilities are added to $F_0$ in this phase.

Then, iteratively for each $p\in\calF\cup\calC'$ s.t. $\sum_{j\in\calC':c_{jp}\leq \delta c_{pF^\star}}(c_{pF^\star}-r_j/(1-\delta))^+>\rho\est$, we set $F_0\leftarrow F_0\cup\{\kappa^\star_p\}$, and remove from $\calC'$ all clients within distance $\delta c_{pF^\star}$ from $p$, that is,
$\calC'\leftarrow\calC'\setminus\{j\in\calC':c_{jp}\leq \delta c_{pF^\star}\}$. For each such removed $j$, using the triangle inequality, the nearest open facility in $F^\star$ is at a distance at least $c_{jF^\star}\geq -c_{jp}+c_{pF^\star}\geq(1-\delta)c_{pF^\star}$, so the total contribution to the true objective of $\scr{I}$ removed is at least
$
\sum_{j\in\calC':c_{jp}\leq\delta c_{pF^\star}}((1-\delta)c_{pF^\star}-r_j)^+>(1-\delta)\rho\est.
$
After this phase, \ref{item:sparse2:knapsack} is satisfied. Since the optimal objective of $\scr{I}$ is at most $\est$, at most $O(1/(\rho-\rho\delta))$ facilities are added to $F_0$ and at most that many closed balls are removed from $\calC'$. This finishes the construction. 

Since $|F_0|=O(1/(\rho-\rho\delta))$ and $\calC'$ is obtained from $\calC$ by removing at most that many closed balls, we can enumerate all possible such procedures and the number of choices is at most $n_0^{O(1/(\rho-\rho\delta))}$. This eliminates the dependence on $F^\star$, and the existence of such an extended instance is guaranteed, thus \ref{item:sparse3:knapsack} follows.

Finally, in the second phase above, for each $j\in\calC\setminus\calC'$ removed by $p$, using the triangle inequality, $c_{jF_0}\leq c_{j\kappa^\star_p}\leq c_{jp}+c_{pF^\star}\leq(1+\delta)c_{pF^\star}\leq\frac{1+\delta}{1-\delta}c_{jF^\star}$. We thus have
$\frac{1-\delta}{1+\delta}\sum_{j\in\calC\setminus\calC'}\left(c_{jF_0}-\frac{1+\delta}{1-\delta}r_j\right)^++\sum_{j\in\calC'}(c_{jF^\star}-r_j)^+
\leq\sum_{j\in\calC}(c_{jF^\star}-r_j)^+\leq\est,$
hence \ref{item:sparse4:knapsack} follows.
\end{proof}


\subsection{The Strengthened Relaxation}

In this section, we proceed with an extended instance $\scr{J}=(\calF,\calC',c,r_j,w_i,W,F_0)$ that satisfies \cref{theorem:sparsifying:knapsack}.
For each $j\in \calC'$, define $R_j$ as the maximum $R\geq0$ that satisfies $\sum_{j'\in\calC':c_{jj'}\leq\delta R}(R-r_{j'}/(1-\delta))^+\leq\rho\est$. Using \ref{item:sparse2:knapsack}, one has $R_j\geq c_{jF^\star}$. The relaxation is as follows.
\begin{alignat*}{3}
    \text{min}\quad&\sum_{j\in\calC'}\sum_{i\in\calF}x_{ij}(c_{ij}-r_j)^+\tag{$\mathrm{LP\text{-}Knap}$}\label{lp:knap:discount}\\
    \text{s.t.}\quad
    &\sum_{i\in\calF}x_{ij}=1\quad\forall j\in\calC'\quad\quad\sum_{i\in\calF}w_iy_i\leq W\\
    &0\leq x_{ij}\leq y_i\leq 1\quad\forall i\in\calF,j\in\calC'\\
    &y_i=1 \quad\forall i\in F_0\tag{$\mathrm{K4}$}\label{lp:knap:discount4}\\
    &x_{ij}=0 \quad\forall i,j,\text{ s.t. }c_{ij}>R_j\tag{$\mathrm{K5}$}\label{lp:knap:discount5}\\
    &x_{ij}=0 \quad\forall i\notin F_0,j,\text{ s.t. }c_{ij}-r_j>\rho\est\tag{$\mathrm{K6}$}\label{lp:knap:discount6}\\
    &\sum_{j\in\calC'}x_{ij}(c_{ij}-r_j)^+\leq\rho\est y_i  \quad\forall i\notin F_0.\tag{$\mathrm{K7}$}\label{lp:knap:discount7}
\end{alignat*}

\begin{lemma}\label{lemma:knap:objective}
\ref{lp:knap:discount} is feasible. Let $(\bar x,\bar y)$ be its optimal solution, then its objective $U$ satisfies
\[
\frac{1-\delta}{1+\delta}\sum_{j\in\calC\setminus\calC'}\left(c_{jF_0}-\frac{1+\delta}{1-\delta}r_j\right)^++U\leq\est.
\]
\end{lemma}
\begin{proof}
Let $(x^{(0)},y^{(0)})$ be the integral solution induced by the optimal solution $F^\star$. It satisfies the first three sets of constraints in \ref{lp:knap:discount} by definition. \cref{lp:knap:discount4} is also satisfied because $F_0\subseteq F^\star$ due to \ref{item:sparse3:knapsack}.

For \cref{lp:knap:discount5}, we consider the definition of $R_j$. Since $\scr{J}$ is $(\rho,\delta,\est)$-sparse w.r.t. $F^\star$, we have $R_j\geq c_{jF^\star}$ using \ref{item:sparse2:knapsack}. This shows that $c_{ij}$ is not a facility-client connection for $j$ when $c_{ij}>R_j\geq c_{jF^\star}$, i.e., $x^{(0)}_{ij}=0$ when $c_{ij}>R_j$, hence \cref{lp:knap:discount5} is satisfied by $(x^{(0)},y^{(0)})$.

For \cref{lp:knap:discount6}, using the sparse property \ref{item:sparse1:knapsack}, any client $j$ connected to $i\notin F_0$ satisfies $(c_{ij}-r_j)^+\leq\rho\est$. \cref{lp:knap:discount7} is also satisfied according to \ref{item:sparse1:knapsack}.
The second assertion of the lemma follows directly from \ref{item:sparse4:knapsack}.
\end{proof}

\paragraph{Prepare the fractional solution} After we solve \ref{lp:knap:discount} and have a fractional solution $(\bar x,\bar y)$, we use the facility duplication technique \cite{charikar2012dependent,krishnaswamy2018constant} to obtain additional properties. We have the following lemma.
\begin{lemma}\label{lemma:facility:duplication}
We can add co-located copies of facilities to $\calF$, create a vector $\hat y\in[0,1]^\calF$, define $F_j\subseteq\{i\in\calF:c_{ij}\leq R_j\}$ for each $j\in\calC'$ such that the following holds.
\begin{enumerate}[label=(\ref{lemma:facility:duplication}.\arabic*), ref=(\ref{lemma:facility:duplication}.\arabic*), leftmargin=1.2cm]
    \item\label{item:fully:connected} Fully connected clients: $\hat y(F_j)=\sum_{i\in F_j}\hat y_i=1$.
    \item\label{item:knapsack:constraint} Knapsack constraint satisfied: $\sum_{i\in\calF}w_i\hat y_i\leq W$.
    \item\label{item:preselected} Pre-selected: $\forall i\in F_0$, $\sum_{i'\text{ co-located with }i}\hat y_{i'}=1$.
    \item\label{item:total:cost} LP objective: $\sum_{j\in\calC'}\sum_{i\in F_j}\hat y_i(c_{ij}-r_j)^+\leq U$.
    \item\label{item:star:cost} Bounded star cost: for each $i$ \emph{not} co-located with $F_0$, $\sum_{j\in\calC':i\in F_j}(c_{ij}-r_j)^+\leq2\rho\est$.
\end{enumerate}
\end{lemma}
\begin{proof}
To avoid confusion, we create a copy $\calF'$ of $\calF$, let $\hat y\in[0,1]^{\calF'}$ with $\hat y\leftarrow\bar y$ initially and set $F_j=\{i\in\calF':\bar x_{ij}>0\}$.
$\hat y$ and $\{F_j\}_{j\in\calC'}$ are always supported on $\calF'$.
We will keep $\calF$ and $(\bar x,\bar y)$ unchanged in the following, and iteratively modify $\calF'$ and $\hat y$ such that the properties are satisfied. For each $i\in\calF'$, define its star cost as $\sum_{j\in\calC':i\in F_j}(c_{ij}-r_j)^+$.

For each $i\in\calF$ and each $j\in\calC'$ such that $\bar x_{ij}>0$, we apply the following steps. Order \emph{all facilities} in $\calF'$  that are co-located with $i$ in non-decreasing order of their current star costs. Choose the first several copies such that the sum of their $\hat y$ values is \emph{exactly} $\bar x_{ij}$. If necessary, split the last chosen $i'$ into two co-located copies $\{i_1',i_2'\}$: $\calF'\leftarrow(\calF'\setminus\{i'\})\cup\{i_1',i_2'\}$, $i_1'$ is chosen with $\hat y_{i_1'}$ set to whichever value that is needed, and $i_2'$ is not chosen with $\hat y_{i_2'}\leftarrow\hat y_{i'}-\hat y_{i_1'}$. Replace all copies that are co-located with $i$ in $F_j$ with the selected copies of $i$. For each other $j'$, if $i'$ is split and $i'\in F_{j'}$, set $F_{j'}\leftarrow (F_{j'}\setminus\{i'\})\cup\{i_1',i_2'\}$.

After the procedures, we replace $\calF$ with $\calF'$ and transfer $\hat y$, $\{F_j\}_{j\in\calC'}$ back such that they are supported on $\calF$. \ref{item:fully:connected} to \ref{item:total:cost} are easy to verify, because the vector $\hat y$ and assignments $F_j$ are the same as $(\bar x,\bar y)$ up to facility duplication (when we make facility copies, the facility weights are also duplicated). It remains to prove \ref{item:star:cost}.

In the beginning, for each $i\in F_j$, because $\bar x_{ij}>0$, one has $(c_{ij}-r_j)^+\leq\rho\est$ by \cref{lp:knap:discount6}. Fix an original facility $i\in\calF$ and let $J_i=\{j\in\calC':\bar x_{ij}>0\}$. W.l.o.g., suppose the clients are $J_i=\{j_1,\dots,j_\ell\}$ such that in our procedures above, we consider $(i,j_1),\dots,(i,j_\ell)$ in order. We need the following claim.
\begin{claim}\label{claim:star:cost:gap}
At any point during the procedures, for $\calF(i)\subseteq\calF'$ being the current set of copies co-located with $i$, the difference between maximum and minimum star costs in $\calF(i)$ is at most $\rho\est$.
\end{claim}
\begin{proof}
After $(i,j_t)$ is considered, one has $i'\in F_{j_s}$ for each $i'\in \calF(i)$ and each $s>t$, since when we split facilities, we add both copies to $F_{j_s}$, and later clients are still not handled. We use induction on the number $t\geq0$ of pairs $(i,j_s),s\in[\ell]$ we \emph{have considered}. 
When $t=0$, the base case is trivial since there is only one copy of $i$.

Suppose the claim holds for $t-1$, $t\geq1$. \emph{Before} we consider $(i,j_t)$, using the induction hypothesis, the difference between the maximum and minimum star costs is at most $\rho\est$. By ordering the copies of $i$ in non-decreasing order of their current star costs, we subtract $(c_{j_ti}-r_{j_t})^+$ from \emph{all} star costs of $\calF(i)$, then add it back to some copies with the smallest star costs. Since $(c_{j_ti}-r_{j_t})^+\leq\rho\est$ using \cref{lp:knap:discount6}, the difference between the maximum and minimum \emph{after} $(i,j_t)$ is still at most $\rho\est$. This proves the induction step.
\end{proof}

In particular, \cref{claim:near:facility} holds after we consider all pairs for $i$. By definition of the procedures, for each $j\in J_i$, the sum of $\hat y$ over $F_j\cap\calF(i)$ is exactly $\bar x_{ij}$, therefore
\begin{align*}
    &\sum_{i'\in\calF(i)}\hat y_{i'}\sum_{j\in J_i:i'\in F_j}(c_{ji}-r_j)^+
    =\sum_{j\in J_i}(c_{ji}-r_j)^+\sum_{i'\in \calF(i):i'\in F_j}\hat y_{i'}\\
    &=\sum_{j\in J_i}\bar x_{ij}(c_{ij}-r_j)^+
    =\sum_{j\in\calC'}\bar x_{ij}(c_{ij}-r_j)^+\leq\rho\est \bar y_i,
\end{align*}
where the last inequality is due to \cref{lp:knap:discount7}. The above inequality implies that
\[
\min_{i'\in\calF(i)}\sum_{j\in J_i:i'\in F_j}(c_{ji}-r_j)^+\leq\frac{\rho\est\bar y_i}{\sum_{i'\in\calF(i)}\hat y_{i'}}=\rho\est.
\]
Using \cref{claim:near:facility}, the maximum star cost among $\calF(i)$ is at most $2\rho\est$ and thus \ref{item:star:cost} holds.
\end{proof}

\subsection{Iterative Rounding}

Using the new fractional solution $\hat y$, we discretize the metric similarly to \kmd and construct an auxiliary LP for rounding, using the knapsack constraint instead of the cardinality constraint. We find a suitable $b\in[0,1)$ as the parameter for metric discretization using the same method as \cref{lemma:iterative:zero} and obtain an initial objective value at most $\frac{\tau-1}{\ln\tau}$ times the original objective of \ref{lp:knap:discount}, that is, at most $\frac{\tau-1}{\ln\tau}U$ according to \ref{item:total:cost}. It is worth noting that, in the extended instance $\scr{J}$, we intend to keep all facilities in $F_0$ open, thus at the beginning of iterative rounding, we add a \emph{virtual client} to $C^\star$ for each $i\in F_0$ that is co-located with it (discount set to 0). The initial radius level of the virtual client is equal to $-1$, hence it will never be removed from $C^\star$, and there is always at least a unit volume of facilities co-located with such $i$. An upshot is that we always have $F_0\subseteq \bar F$ in the final solution $\bar F$, as can be seen below.

We use the same iterative rounding algorithm \cref{algo:iterative} with step size $h=1$. When the algorithm returns an optimal basic solution $y'\in[0,1]^\calF$, the tight constraints at $y'$ belong to a single laminar family and a knapsack constraint. Suppose there are $t$ strictly fractional variables in $y'$. Since $y'$ is a basic solution, there exist at least $t$ non-trivial (that is, not in the form of $y_i\geq0$ or $y_i\leq1$) constraints that are linearly independent. But at most $t/2$ such tight constraints can be found in the laminar family, hence $t\leq t/2+1$ and we have $t\in\{0,1,2\}$.

\begin{lemma}\label{lemma:iterative:knap}
After \cref{algo:iterative} returns $y'$, its objective is at most that of $\hat y$ in the original auxiliary relaxation. For each $j\in \calC'$, one has $
\sum_{i\in\calF:c_{ij}\leq\frac{3\tau-1}{\tau-1}D_{\ell_j}}y'_i\geq1.
$
\end{lemma}
\begin{proof}
Using the tight constraints for $C^\star$, the arguments are the same as \cref{lemma:iterative:three}, except that $y'$ is not necessarily integral, hence we omit the proof here.
\end{proof}
Using \cref{lemma:iterative:knap} and the same argument as \cref{theorem:main:iterative:k}, if we connect each client $j\in\calC'$, fractionally, to the nearest unit volume of facilities in $y'$,
the following discounted objective where we subtract the scaled discount $\frac{\tau(3\tau-1)}{\tau-1}r_j$, 
is at most $\frac{3\tau-1}{\tau-1}\cdot\frac{\tau-1}{\ln\tau}U=\frac{3\tau-1}{\ln\tau}U$,
\begin{equation}
\sum_{j\in\calC'}\left[\sum_{i\in B_j}y'_i\left(c_{ij}-\frac{\tau(3\tau-1)}{\tau-1}r_j\right)^++(1-y'(B_j))\left(\frac{3\tau-1}{\tau-1}D_{\ell_j}-\frac{\tau(3\tau-1)}{\tau-1}r_j\right)^+\right].\label{eqn:final:objective1}
\end{equation}

\paragraph{Open integral facilities}
If $t=0$, set $y^\star\leftarrow y'$ and we are done; if $t=1$ and $y'_{i_2}\in(0,1)$. We let $y^\star$ be the same as $y'$ except that $y^\star_{i_2}=0$; if $t=2$, suppose $y'_{i_1},y'_{i_2}\in(0,1)$ and w.l.o.g. $w_{i_1}\leq w_{i_2}$. Using the previous argument on basic solutions, we must have $y'_{i_1}+y'_{i_2}=1$. We let $y^\star$ be the same as $y'$ except that $y^\star_{i_1}=1$, $y^\star_{i_2}=0$.
Define $\bar F=\{i\in\calF:y^\star_i=1\}$, and remove co-located facilities if there are any (again, this does not affect \cref{lemma:iterative:knap}). It is easy to verify that in each case, $w(\bar F)\leq W$ holds.

We consider \emph{final} outer balls $F_j$, radius levels $\ell_j$ and inner balls $B_j$ in the following.
For $j\in\calC'$ such that $i_2\in F_j\setminus B_j$, closing $i_2$ does not change the contribution of $j$ in \cref{eqn:final:objective1} since $i_2\notin B_j$. This leaves only the cases where $i_2\in B_j$ and \emph{not} co-located with $F_0$.
Let $J=\{j\in\calC':i_2\in B_j\}$, $i'\in\bar F$ be the nearest open facility to $i_2$, and $d=c_{i_2i'}$.
We consider the following two cases, where $\gamma>0$ is a parameter to be determined.
\begin{description}
\item[$J_1=\{j\in J:c_{ji_2}\geq\gamma d\}$.]
We have $i_2\in F_j$ before iterative rounding for each $j\in J_1$, since each $F_j$ only shrinks during \cref{algo:iterative}. Using \ref{item:star:cost}, one has $\sum_{j\in J_1}(c_{ji_2}-r_j)^+\leq2\rho\est$. Using triangle inequality, $c_{ji'}\leq c_{ji_2}+c_{i_2i'}\leq(1+1/\gamma)c_{ji_2}$, which shows
\begin{equation}
\sum_{j\in J_1}\left(c_{ji'}-\left(1+\frac{1}{\gamma}\right)r_j\right)^+\leq\frac{1+\gamma}{\gamma}2\rho\est.\label{eqn:reroute:1}
\end{equation}
\item[$J_2=J\setminus J_1=\{j\in J:c_{ji_2}<\gamma d\}$.] Fix any $j\in J_2$. Let $\sigma=\frac{\tau(3\tau-1)}{\tau-1}$ for convenience. Using \cref{lemma:iterative:knap} and triangle inequality, $c_{i'i_2}-c_{ji_2}$ is at most the distance from $j$ to its nearest open facility in $\bar F$, thus $d-\gamma d\leq c_{i'i_2}-c_{ji_2}\leq\frac{3\tau-1}{\tau-1}D_{\ell_j}\leq\sigma R_j$. The last inequality is due to $D_{\ell_j}\leq \tau R_j$, and thus $R_j\geq(1-\gamma)d/\sigma$. Suppose for now that the choice of $\gamma$ is such that $2\gamma d\leq \delta R_j$. It follows from triangle inequality that $c_{jj'}<2\gamma d$ for each $j'\in J_2$, therefore $J_2\subseteq\{j'\in\calC':c_{jj'}\leq\delta R_j\}$. For each $j'\in J_2$, we have $c_{j'i'}\leq(1+\gamma)d\leq\frac{\sigma(1+\gamma)}{1-\gamma}R_j\defeq \eta_{\tau,\gamma}R_j$.
By definition of $R_j$,
\begin{align}
    &\sum_{j'\in J_2}(c_{j'i'}-\frac{\eta_{\tau,\gamma}}{1-\delta}r_{j'})^+
    \leq\sum_{j'\in\calC':c_{jj'}\leq\delta R_j}(\eta_{\tau,\gamma}R_j-\frac{\eta_{\tau,\gamma}}{1-\delta}r_{j'})^+\notag\\
    &=\eta_{\tau,\gamma}\sum_{j'\in\calC':c_{jj'}\leq\delta R_j}(R_j-\frac{1}{1-\delta}r_{j'})^+
    \leq\eta_{\tau,\gamma}\rho\est.\label{eqn:reroute:2}
\end{align}
To make sure $2\gamma d\leq \delta R_j$, we let $(1-\gamma)/\sigma=2\gamma/\delta$, thus $\gamma=\delta/(2\sigma+\delta)$. As a result, we have $\eta_{\tau,\gamma}=\delta+\sigma$ and $1+1/\gamma=2\sigma/\delta+2$.
\end{description}

\paragraph{Conclusion}
We let $\delta=2/3$.
In \cref{eqn:final:objective1}, the fractional solution has objective at most $\frac{3\tau-1}{\ln\tau}U$ against discounts $\sigma r_j$. By changing the underlying solution to an integral one $\bar F$, and bounding the ``rerouting'' cost in \cref{eqn:reroute:1}, \cref{eqn:reroute:2}, the clients in $J$ additionally contribute at most $\rho(7\sigma+14/3)\est$ to the overall objective, by considering discounts $(3\sigma+2)r_j$.

Finally, since we use virtual clients in iterative rounding and keep $F_0\subseteq \bar F$, by connecting each $j\in\calC\setminus\calC'$ to its nearest facility in $\bar F$ and using \cref{lemma:knap:objective}, the clients in $\calC\setminus\calC'$ have a total cost
$\sum_{j\in\calC\setminus\calC'}(c_{j\bar F}-5r_j)^+\leq5(\est-U)$.
This shows that the total objective of $\bar F$, with discounts $\max\{\sigma,3\sigma+2,5\}r_j$, is at most $\big(\max\{5,\frac{3\tau-1}{\ln\tau}\}+\rho(7\sigma+14/3)\big)\est$. Recall $\est\leq(1+\epsilon)\opt_2$. We let $\epsilon$ and $\rho$ be small constants and thus have the following result.
\begin{theorem}\label{theorem:main:iterative:knapsack}
Let $\alpha_\tau''=\frac{3\tau(3\tau-1)}{\tau-1}+2$ and $\beta_\tau''=\frac{3\tau-1}{\ln\tau}$. There exists a bi-criteria $(\alpha_\tau'',\beta_\tau''+O(\rho+\epsilon))$-approximation algorithm for \knapmd for any $\tau>1$ and $\epsilon,\rho>0$ in running time $n_0^{O(1/\rho)}/\epsilon$.
\end{theorem}
\section{Application}\label{section:application}

\subsection{Stochastic Center Clustering}

In this section, we consider the unassigned stochastic center clustering problems. The input contains a set of stochastic points $V$, and each of them follows a known distribution and independently realizes at a random location in $\calC$. For a set of open facilities $S\subseteq\calF$, the objective is the expectation of maximum distance from each realized point to its nearest open facility in $S$, that is, $\E_\tau[\max_{v\in V}c_{\tau(v)S}]$ where $\tau:V\rightarrow\calC$ is the random realization map such that $v\in V$ realizes at $\tau(v)\in\calC$.
In the cardinality version called unassigned stochastic $k$-center (\unkcen), the chosen facilities $S$ must satisfy $|S|\leq k$, In the matroid version called unassigned stochastic matroid center (\unmatcen), $S$ is required to be an independent set in a given matroid $\calM=(\calF,\calI)$. In the knapsack version called unassigned stochastic knapsack center (\unknapcen), $S$ needs to have a combined weight of at most $W$, i.e., $w(S)=\sum_{i\in S}w_i\leq W$. Guha and Munagala~\cite{guha2009exceeding} give a $(45+\epsilon)$-approximation algorithm for \unkcen.

Plugging the implicit bi-criteria $(5,5+\epsilon)$-approximation for \unikmd \cite{chakrabarty2019approximation} into the algorithm in Section~3.4 of \cite{guha2009exceeding}, one obtains an improved $(30+\epsilon)$-approximation for \unkcen.
Similarly, using our bi-criteria approximations for \matmd and \knapmd, we obtain the first $O(1)$-approximations for \unmatcen and \unknapcen. We consider \unmatcen for example here and give a proof sketch below.

We first need a few useful lemmas. The first one is due to Kleinberg~\etalcite{kleinberg2000allocating} and we prove a weaker version here for completeness~\cite{guha2009exceeding}.
\begin{lemma}\label{lemma:stochastic:1} (\cite{guha2009exceeding,kleinberg2000allocating}) Given $T>0$, $n$ independent Bernoulli variables $X_i,i\in[n]$ where $\E[X_i]=p_i$, and $n$ real numbers $s_i\geq T$, suppose $\E[\max_i s_iX_i]<T/3$, then one has $\sum_is_ip_i<T$.
\end{lemma}
\begin{proof}
If $\sum_ip_i\geq2/3$, since the variables are independent, we have
\[
\E[\max_iX_i]=\Pr[\exists i,X_i=1]=1-\prod_i(1-p_i)\geq1-\exp\left(-\sum_ip_i\right)
\geq1-\exp(-2/3)>1/3,
\]
which leads to $\E[\max_is_iX_i]\geq T\E[\max_iX_i]>T/3$, contradicting the given condition. Thus we have $\sum_ip_i<2/3$.

For each $i$, let $\calE_i$ denote the event that $X_i=1$ and $X_j=0$ for each $j\neq i$. These events are mutually exclusive. Using the fact that the variables are independent, we have
\begin{align*}
\Pr[\calE_i]&=\Pr[X_i=1]\Pr[\forall j\neq i,X_j=0]=p_i\Pr\left[\sum_{j\neq i}X_j=0\right]\\
&=p_i\left(1-\Pr\left[\sum_{j\neq i}X_j\geq1\right]\right)
\geq p_i\left(1-\E\left[
\sum_{j\neq i}X_j
\right]\right)=p_i\left(1-\sum_{j\neq i}p_j\right)\geq p_i/3,
\end{align*}
where we use $\sum_ip_i<2/3$ in the last inequality. Since the events are mutually exclusive and $\calE_i$ implies $\max_is_iX_i=s_i$, the following inequalities shows the lemma.
\[
\frac{T}{3}>\E[\max_is_iX_i]\geq\E\left[\sum_is_i\mathbbm{1}[\calE_i]\right]\geq\sum_is_ip_i/3.\tag*{\qedhere}
\]
\end{proof}

We consider a slightly more general version of uniform \matmd here, where the clients are $\calC$ with uniform discounts $T$. Given any set of facilities $S$, the contribution of each $j\in\calC$ to the objective is $p_j(c_{jS}-T)^+$, where $p_j$ is the probability of there existing a point $v\in V$ realizing at $j$, i.e., $p_j=\Pr_\tau[\exists v\in V,\tau(v)=j]$ for the random realization map $\tau$. $\{p_j:j\in\calC\}$ is easy to compute and one notices that the algorithm in \cref{theorem:main:iterative:matroid} provides the same approximation guarantee by simply modifying the LP relaxations. For such a \matmd instance with uniform discounts $T\geq0$, we let $\opt_T$ denote the optimum. We also let $\opt^\star$ be the optimum to the original \unmatcen instance.

\begin{lemma}\label{lemma:stochastic:2} (\cite{guha1998greedy})
$\opt_0\geq\opt^\star\geq\frac{1}{3}\max\{T:\opt_T\geq T\}$.
\end{lemma}
\begin{proof}
Let $S\in\calI$ be any feasible solution to the \unmatcen instance and thus 
\begin{align*}
\opt^\star&\leq\E_\tau\left[\max_{v\in V}c_{\tau(v)S}\right]
=\E_\tau\left[\max_{j\in\calC}\mathbbm{1}[\exists v,\tau(v)=j]\cdot c_{jS}\right]
\leq\E_\tau\left[\sum_{j\in\calC}\mathbbm{1}[\exists v,\tau(v)=j]\cdot c_{jS}\right]\\
&\leq\sum_{j\in\calC}\Pr_\tau[\exists v,\tau(v)=j]\cdot c_{jS}=\sum_{j\in\calC}p_jc_{jS}.
\end{align*}
For $T=0$, the minimum of the RHS over all feasible $S\in\calI$ is exactly $\opt_0$, thus the first inequality follows.
We now show if $\opt_T\geq T$, we have $\opt^\star\geq T/3$. For the sake of contradiction, suppose $\opt^\star<T/3$. To this end, let $\calL_T:(\calC\cup\calF)^2\rightarrow\mathbb{R}_{\geq0}$ be a new function such that $\calL_T(i,j)=c_{ij}$ if $c_{ij}\geq T$ and 0 otherwise. It is easy to see that $\calL_T(i,j)\geq(c_{ij}-T)^+$ always holds.

Let $S^\star$ be the optimal solution to the original \unmatcen instance, thus we obtain
\[
\frac{T}{3}>\opt^\star=\E_\tau\left[\max_{j\in\calC}\mathbbm{1}[\exists v,\tau(v)=j]\cdot c_{jS^\star}\right]
\geq\E_\tau\left[\max_{j\in\calC}\mathbbm{1}[\exists v,\tau(v)=j]\cdot\calL_T(j,S^\star)\right].
\]
We invoke \cref{lemma:stochastic:1} on clients $j$ such that $\calL_T(j,S^\star)=c_{jS^\star}\geq T$, and obtain 
\[\sum_{j:c_{jS^\star}\geq T}p_j\calL_T(j,S^\star)<T,\]
which leads to 
\[
T>\sum_{j:c_{jS^\star}\geq T}p_j\calL_T(j,S^\star)
=\sum_{j\in\calC}p_j\calL_T(j,S^\star)
\geq\sum_{j\in\calC}p_j(c_{jS^\star}-T)^+\geq\opt_T,
\]
contradicting $\opt_T\geq T$. Therefore, the initial assumption is false and we have $\opt^\star\geq T/3$.
\end{proof}

Following \cite{guha2009exceeding}, we repeatedly solve \matmd instances with a decreasing uniform discount.
We start with $T=\max_{i,j}c_{ij}$ and scale it by $T\leftarrow (1-\epsilon)T$ after each iteration. In each iteration, we use \cref{theorem:main:iterative:matroid} and approximately solve the \matmd instance with uniform discounts $T$ (and scaled by the probabilities). Recall that in the bi-criteria $(\atau',\btau')$-approximation for \matmd, we have $\atau'=\frac{\tau(3\tau-1)}{\tau-1}$ and $\btau'=\frac{3\tau-1}{\ln\tau}$.

Let $T^\star$ be the smallest value such that the output solution $\baring{S}$ obtains objective $\leq\btau'T^\star$ against discounts $\atau'T^\star$. Now, consider the performance of $\baring{S}$ in the \unmatcen instance, with objective
\begin{align*}
\E_\tau\left[\max_{v\in V}c_{\tau(v)\baring{S}}\right]
&\leq\atau'T^\star+\E_\tau\left[\max_{v\in V}(c_{\tau(v)\baring{S}}-\atau'T^\star)^+\right]\\
&\leq \atau'T^\star+\E_\tau\left[\sum_{j\in \calC}\mathbbm{1}[\exists v\in V,\tau(v)=j]\cdot(c_{j\baring{S}}-\atau'T^\star)^+\right]\\
&\leq \atau'T^\star+\sum_jp_j(c_{j\baring{S}}-\atau'T^\star)^+\leq \atau'T^\star+\btau'T^\star.
\end{align*}

On the other hand, because $T^\star$ is the smallest such value, when the uniform discounts are $T'=(1-\epsilon)T^\star$ in the next iteration, the obtained objective of \matmd is $>\btau'T'$, which shows that the optimum $\opt_{T'}$ of \matmd with uniform discounts $T'$ satisfies $\opt_{T'}>T'$, using~\cref{theorem:main:iterative:matroid}. According to \cref{lemma:stochastic:2}, this shows that $T'\leq3\opt^\star$, whence we conclude that $\baring{S}$ has objective (in the \unmatcen instance) at most $(\atau'+\btau')T^\star=(\atau'+\btau')T'/(1-\epsilon)\leq3(1+2\epsilon)(\atau'+\btau')\opt^\star$. Let $\tau=1.985$, and we obtain a $(51.638+\epsilon)$-approximation for \unmatcen.
A similar $(117.263+\epsilon)$-approximation algorithm can also be obtained in the knapsack case.


\subsection{Universal Matroid Median}

Following the recently-developed framework by Ganesh~\etalcite{ganesh2021universal}, we obtain bi-criteria constant-factor approximations for universal matroid median and its variant with fixed clients. We remark that the knapsack versions of these two problems remain interestingly open, as the unbounded integrality gap of the relaxation with knapsack constraints is hard to overcome and the current methods fail to provide any similar guarantee. We provide the algorithm for universal matroid median here, and defer the proof sketch for universal matroid median with fixed clients to \cref{section:universal:fixed}.

\begin{theorem}\label{theorem:matroid:universal}
There exists a bi-criteria $(O(1),O(1))$-approximation algorithm for the universal matroid median problem.
\end{theorem}

\subsubsection{Fractional Solution}
In this section, we consider the universal matroid median problem, where we are given clients $\calC$, facilities $\calF$, a matroid $\calM=(\calF,\calI)$ with rank function $r_\calM:2^\calF\rightarrow\mathbb{Z}_{\geq0}$, and the goal is find a set of facilities $S\in\calI$ with minimum regret, where the regret of $S$ is defined as the maximum difference between its objective value and the instance-optimum, over the choices of realized clients being an arbitrary $C_1\subseteq\calC$. More specifically, we define the regret of $S$
\begin{equation}
    \regret(S)\defeq\max_{C_1\subseteq\calC}\{\solution(C_1,S)-\optimal(C_1)\},\label{eqn:regret:definition}
\end{equation}
\begin{equation*}
    \textrm{where}\;\solution(C_1,S)\defeq\sum_{j\in C_1}c_{jS},\,
    \optimal(C_1)\defeq\min_{S\in\calI}\solution(C_1,S),\,
\end{equation*}
and we want to find $S\in\calI$ that attains the minimum $\minregret\defeq\min_{S\in\calI}\regret(S)$.

For universal $k$-median which has a weaker $k$-cardinality constraint on the set of facilities we can choose, Ganesh~\etalcite{ganesh2021universal} devise a bi-criteria $(O(1),O(1))$-approximation algorithm, that is, an algorithm that outputs $S\subseteq\calF,\,|S|\leq k$ s.t. $\solution(C_1,S)\leq\alpha\cdot\optimal(C_1)+\beta\cdot\minregret$ for each $C_1\subseteq\calC$.
They also show that even for universal $k$-median, one can only hope for bi-criteria $(\alpha,\beta)$-approximations where $\alpha,\beta$ are both bounded away from 1 for general metrics unless $\mathrm{P=NP}$. Therefore, we also aim at obtaining bi-criteria constant approximations for universal matroid median.

We first have the following relaxation with exponentially many constraints.
\begin{alignat}{3}
\text{min}\quad&&r&\geq0\tag{$\mathrm{M\text{-}UNI}$}\label{lp:universal:matroid}\\
\text{s.t.}\quad&&y(S)&\leq r_\calM(S)&&\forall S\subseteq\calF\notag\\
&& 0\leq x_{ij}&\leq y_i&&\forall i,j\notag\\
&&\sum_{i\in\calF}x_{ij}&=1&&\forall j\in\calC\notag\\
&&\sum_{j\in C_1}\sum_{i\in\calF}x_{ij}c_{ij}&\leq\optimal(C_1)+r\quad&&\forall C_1\subseteq\calC.\label{lp:universal:matroid:con4}
\end{alignat}

We notice that the minimum regret solution $\minregretsol\defeq\argmin_{S\in\calI}\regret(S)$ induces a feasible solution to \ref{lp:universal:matroid}, so its optimal objective is at most $\minregret$ by definition of $\minregret$ and constraint \eqref{lp:universal:matroid:con4}.
We want to use the ellipsoid algorithm to solve \ref{lp:universal:matroid}.
Unfortunately, we cannot find the value of any $\optimal(C_1)$ in polynomial time since this would imply a solution to standard $k$-median which is NP-hard, and we have to resort to approximate separating oracles. To separate a candidate solution $(x,y,r)$ to \ref{lp:universal:matroid}, it is equivalent to check whether the following holds,
\begin{equation}
    \max_{C_1\subseteq\calC}\max_{S'\in\calI}\left[
    \sum_{j\in C_1}\left(-c_{jS'}+\sum_{i\in\calF}x_{ij}c_{ij}\right)
    \right]\leq r,\label{eqn:oracle1}
\end{equation}
which, by defining 
\begin{equation}
f_{x,y}(S')\defeq
\sum_{j\in\calC}\left(-c_{jS'}+\sum_{i\in\calF}x_{ij}c_{ij}\right)^+,\label{eqn:fxy}
\end{equation}
is equivalent to deciding whether $\max_{S'\in\calI}f_{x,y}(S')\leq r$. Using the observation in~\cite{ganesh2021universal}, for each fixed $(x,y)$, $f_{x,y}:2^\calF\rightarrow\mathbb{R}_+$ is the sum of $|\calC|$ monotone submodular functions, hence monotone submodular itself. Thus a simple greedy algorithm reveals a 2-approximate solution to the problem of maximizing $f_{x,y}(S')$ subject to $S'\in\calI$, according to a classic result in~\cite{nemhauser1978analysis}. Further, using a standard result on ellipsoid methods~\cite{grotschel1981ellipsoid}, we can approximately solve \ref{lp:universal:matroid} and obtain a fractional solution $(\bar x,\bar y,\bar r)$ with actual regret $\max_{S'\in\calI}f_{\bar x,\bar y}(S')\leq2\minregret$. We let $\fav(j)=\sum_{i\in\calF}\bar x_{ij}c_{ij}$ be the cost of client $j$ in the fractional solution.

\subsubsection{Rounding and Analysis}

The standard relaxation for matroid median has an integrality gap of at most 8~\cite{swamy2016improved}, thus for the fractional solution, we have
\begin{equation}
    \forall C_1\subseteq\calC,\,\frac{1}{8}\optimal(C_1)\leq\sum_{j\in C_1}\fav(j)\leq\optimal(C_1)+2\minregret.\label{eqn:sandwich}
\end{equation}

We then solve an instance of \matmd, where $r_j=8\fav(j)$ for each $j\in\calC$. Using the aforementioned $(\atau',\btau')$-approximation in \cref{theorem:main:iterative:matroid}, we obtain a solution $\baring{S}\in\calI$ s.t.
\begin{align}
    \sum_{j\in\calC}\left(c_{j\baring{S}}-8\atau'\fav(j)\right)^+&\leq \btau'\sum_{j\in\calC}(m_j-8\fav(j))^+=\btau'\sum_{\substack{j\in\calC\\m_j\geq8\fav(j)}}(m_j-8\fav(j))\notag\\
    &\leq\btau'\left(\sum_{\substack{j\in\calC\\m_j\geq8\fav(j)}}m_j-\optimal(\{j\in\calC:m_j\geq8\fav(j)\})\right)\leq\btau'\minregret,\label{eqn:discount:regret}
\end{align}
where we use \eqref{eqn:sandwich} and $m_j\geq0$ is the connection distance of $j$ in the \emph{minimum regret solution}. This shows that for any $C_1\subseteq\calC$, the total cost of $\baring{S}$ is
\begin{align}
    \sum_{j\in C_1}c_{j\baring{S}}&\leq\sum_{j\in C_1}\left((c_{j\baring{S}}-8\atau'\fav(j))^++8\atau'\fav(j)\right)
    \leq\btau'\minregret+8\atau'\sum_{j\in C_1}\fav(j)\notag\\
    &\leq8\atau'(\optimal(C_1)+2\minregret)+\btau'\minregret
    =8\atau'\optimal(C_1)+(16\atau'+\btau')\minregret,\label{eqn:universal:bicriteria}
\end{align}
where we use \eqref{eqn:sandwich} and \eqref{eqn:discount:regret}. Recall that $\atau'=\frac{\tau(3\tau-1)}{\tau-1}$ and $\btau'=\frac{3\tau-1}{\ln\tau}$, therefore \eqref{eqn:universal:bicriteria} provides a trade-off between the approximation factors. For instance, one can choose $\tau=1.816$ and obtain ratios $<(79.192,165.839)$ minimizing the first factor, or $\tau=1.832$ and obtain ratios $<(79.199,165.824)$ minimizing the second factor.

\bibliographystyle{plainurl}
\bibliography{references.bib}

\begin{thebibliography}{10}

\bibitem{alipour2021improvements}
Sharareh Alipour.
\newblock Improvements on approximation algorithms for clustering probabilistic
  data.
\newblock {\em Knowl. Inf. Syst.}, 63(10):2719--2740, 2021.
\newblock \href {https://doi.org/10.1007/s10115-021-01601-4}
  {\path{doi:10.1007/s10115-021-01601-4}}.

\bibitem{alipour2018improvement}
Sharareh Alipour and Amir Jafari.
\newblock Improvements on the $k$-center problem for uncertain data.
\newblock In {\em Proceedings of the 37th {ACM} {SIGMOD-SIGACT-SIGAI} Symposium
  on Principles of Database Systems}, pages 425--433. {ACM}, 2018.
\newblock \href {https://doi.org/10.1145/3196959.3196969}
  {\path{doi:10.1145/3196959.3196969}}.

\bibitem{arya2004local}
Vijay Arya, Naveen Garg, Rohit Khandekar, Adam Meyerson, Kamesh Munagala, and
  Vinayaka Pandit.
\newblock Local search heuristics for $k$-median and facility location
  problems.
\newblock {\em {SIAM} J. Comput.}, 33(3):544--562, 2004.
\newblock \href {https://doi.org/10.1137/S0097539702416402}
  {\path{doi:10.1137/S0097539702416402}}.

\bibitem{byrka2017improved}
Jaroslaw Byrka, Thomas~W. Pensyl, Bartosz Rybicki, Aravind Srinivasan, and Khoa
  Trinh.
\newblock An improved approximation for $k$-median and positive correlation in
  budgeted optimization.
\newblock {\em {ACM} Trans. Algorithms}, 13(2):23:1--23:31, 2017.
\newblock \href {https://doi.org/10.1145/2981561} {\path{doi:10.1145/2981561}}.

\bibitem{chakrabarty2018generalized}
Deeparnab Chakrabarty and Maryam Negahbani.
\newblock Generalized center problems with outliers.
\newblock In {\em 45th International Colloquium on Automata, Languages, and
  Programming}, volume 107 of {\em LIPIcs}, pages 30:1--30:14, 2018.
\newblock \href {https://doi.org/10.4230/LIPIcs.ICALP.2018.30}
  {\path{doi:10.4230/LIPIcs.ICALP.2018.30}}.

\bibitem{chakrabarty2019approximation}
Deeparnab Chakrabarty and Chaitanya Swamy.
\newblock Approximation algorithms for minimum norm and ordered optimization
  problems.
\newblock In {\em Proceedings of the 51st Annual {ACM} {SIGACT} Symposium on
  Theory of Computing}, pages 126--137, 2019.
\newblock \href {https://doi.org/10.1145/3313276.3316322}
  {\path{doi:10.1145/3313276.3316322}}.

\bibitem{charikar2002constant}
Moses Charikar, Sudipto Guha, {\'{E}}va Tardos, and David~B. Shmoys.
\newblock A constant-factor approximation algorithm for the $k$-median problem.
\newblock {\em Journal of Computer and System Sciences}, 65(1):129--149, 2002.
\newblock \href {https://doi.org/10.1006/jcss.2002.1882}
  {\path{doi:10.1006/jcss.2002.1882}}.

\bibitem{charikar2012dependent}
Moses Charikar and Shi Li.
\newblock A dependent {LP}-rounding approach for the $k$-median problem.
\newblock In {\em Automata, Languages, and Programming - 39th International
  Colloquium, Proceedings, Part {I}}, volume 7391 of {\em Lecture Notes in
  Computer Science}, pages 194--205, 2012.
\newblock \href {https://doi.org/10.1007/978-3-642-31594-7\_17}
  {\path{doi:10.1007/978-3-642-31594-7\_17}}.

\bibitem{chen2016matroid}
Danny~Z. Chen, Jian Li, Hongyu Liang, and Haitao Wang.
\newblock Matroid and knapsack center problems.
\newblock {\em Algorithmica}, 75(1):27--52, 2016.
\newblock \href {https://doi.org/10.1007/s00453-015-0010-1}
  {\path{doi:10.1007/s00453-015-0010-1}}.

\bibitem{cormode2008approximation}
Graham Cormode and Andrew McGregor.
\newblock Approximation algorithms for clustering uncertain data.
\newblock In {\em Proceedings of the Twenty-Seventh {ACM}
  {SIGMOD-SIGACT-SIGART} Symposium on Principles of Database Systems}, pages
  191--200. {ACM}, 2008.
\newblock \href {https://doi.org/10.1145/1376916.1376944}
  {\path{doi:10.1145/1376916.1376944}}.

\bibitem{edmonds2001submodular}
Jack~R. Edmonds.
\newblock Submodular functions, matroids, and certain polyhedra.
\newblock In {\em Combinatorial Optimization - Eureka, You Shrink!}, volume
  2570 of {\em Lecture Notes in Computer Science}, pages 11--26. Springer,
  2001.
\newblock \href {https://doi.org/10.1007/3-540-36478-1\_2}
  {\path{doi:10.1007/3-540-36478-1\_2}}.

\bibitem{ganesh2021universal}
Arun Ganesh, Bruce~M. Maggs, and Debmalya Panigrahi.
\newblock Universal algorithms for clustering problems.
\newblock In {\em 48th International Colloquium on Automata, Languages, and
  Programming}, volume 198 of {\em LIPIcs}, pages 70:1--70:20, 2021.
\newblock \href {https://doi.org/10.4230/LIPIcs.ICALP.2021.70}
  {\path{doi:10.4230/LIPIcs.ICALP.2021.70}}.

\bibitem{gonzalez1985clustering}
Teofilo~F. Gonzalez.
\newblock Clustering to minimize the maximum intercluster distance.
\newblock {\em Theoretical Computer Science}, 38:293--306, 1985.
\newblock \href {https://doi.org/10.1016/0304-3975(85)90224-5}
  {\path{doi:10.1016/0304-3975(85)90224-5}}.

\bibitem{grotschel1981ellipsoid}
Martin Gr{\"{o}}tschel, L{\'{a}}szl{\'{o}} Lov{\'{a}}sz, and Alexander
  Schrijver.
\newblock The ellipsoid method and its consequences in combinatorial
  optimization.
\newblock {\em Comb.}, 1(2):169--197, 1981.
\newblock \href {https://doi.org/10.1007/BF02579273}
  {\path{doi:10.1007/BF02579273}}.

\bibitem{guha1998greedy}
Sudipto Guha and Samir Khuller.
\newblock Greedy strikes back: Improved facility location algorithms.
\newblock In {\em Proceedings of the Ninth Annual {ACM-SIAM} Symposium on
  Discrete Algorithms}, pages 649--657, 1998.
\newblock URL: \url{http://dl.acm.org/citation.cfm?id=314613.315037}.

\bibitem{guha2009exceeding}
Sudipto Guha and Kamesh Munagala.
\newblock Exceeding expectations and clustering uncertain data.
\newblock In {\em Proceedings of the Twenty-Eigth {ACM} {SIGMOD-SIGACT-SIGART}
  Symposium on Principles of Database Systems}, pages 269--278. {ACM}, 2009.
\newblock \href {https://doi.org/10.1145/1559795.1559836}
  {\path{doi:10.1145/1559795.1559836}}.

\bibitem{gupta2021structual}
Anupam Gupta, Benjamin Moseley, and Rudy Zhou.
\newblock Structural iterative rounding for generalized $k$-median problems.
\newblock In {\em 48th International Colloquium on Automata, Languages, and
  Programming}, volume 198 of {\em LIPIcs}, pages 77:1--77:18. Schloss Dagstuhl
  - Leibniz-Zentrum f{\"{u}}r Informatik, 2021.
\newblock \href {https://doi.org/10.4230/LIPIcs.ICALP.2021.77}
  {\path{doi:10.4230/LIPIcs.ICALP.2021.77}}.

\bibitem{harris2019lottery}
David~G. Harris, Thomas~W. Pensyl, Aravind Srinivasan, and Khoa Trinh.
\newblock A lottery model for center-type problems with outliers.
\newblock {\em {ACM} Transactions on Algorithms}, 15(3):36:1--36:25, 2019.
\newblock \href {https://doi.org/10.1145/3311953} {\path{doi:10.1145/3311953}}.

\bibitem{hochbaum1985best}
Dorit~S. Hochbaum and David~B. Shmoys.
\newblock A best possible heuristic for the $k$-center problem.
\newblock {\em Math. Oper. Res.}, 10(2):180--184, 1985.
\newblock \href {https://doi.org/10.1287/moor.10.2.180}
  {\path{doi:10.1287/moor.10.2.180}}.

\bibitem{hochbaum1986unified}
Dorit~S. Hochbaum and David~B. Shmoys.
\newblock A unified approach to approximation algorithms for bottleneck
  problems.
\newblock {\em Journal of the {ACM}}, 33(3):533--550, 1986.
\newblock \href {https://doi.org/10.1145/5925.5933}
  {\path{doi:10.1145/5925.5933}}.

\bibitem{hsu1979bottleneck}
Wen{-}Lian Hsu and George~L. Nemhauser.
\newblock Easy and hard bottleneck location problems.
\newblock {\em Discrete Applied Mathematics}, 1(3):209--215, 1979.
\newblock \href {https://doi.org/10.1016/0166-218X(79)90044-1}
  {\path{doi:10.1016/0166-218X(79)90044-1}}.

\bibitem{huang2017stochastic}
Lingxiao Huang and Jian Li.
\newblock Stochastic $k$-center and $j$-flat-center problems.
\newblock In {\em Proceedings of the Twenty-Eighth Annual {ACM-SIAM} Symposium
  on Discrete Algorithms}, pages 110--129. {SIAM}, 2017.
\newblock \href {https://doi.org/10.1137/1.9781611974782.8}
  {\path{doi:10.1137/1.9781611974782.8}}.

\bibitem{jain2002greedy}
Kamal Jain, Mohammad Mahdian, and Amin Saberi.
\newblock A new greedy approach for facility location problems.
\newblock In {\em Proceedings on 34th Annual {ACM} Symposium on Theory of
  Computing}, pages 731--740, 2002.
\newblock \href {https://doi.org/10.1145/509907.510012}
  {\path{doi:10.1145/509907.510012}}.

\bibitem{jain2001approximation}
Kamal Jain and Vijay~V. Vazirani.
\newblock Approximation algorithms for metric facility location and $k$-median
  problems using the primal-dual schema and lagrangian relaxation.
\newblock {\em J. {ACM}}, 48(2):274--296, 2001.
\newblock \href {https://doi.org/10.1145/375827.375845}
  {\path{doi:10.1145/375827.375845}}.

\bibitem{kleinberg2000allocating}
Jon~M. Kleinberg, Yuval Rabani, and {\'{E}}va Tardos.
\newblock Allocating bandwidth for bursty connections.
\newblock {\em {SIAM} J. Comput.}, 30(1):191--217, 2000.
\newblock \href {https://doi.org/10.1137/S0097539797329142}
  {\path{doi:10.1137/S0097539797329142}}.

\bibitem{krishnaswamy2011matroid}
Ravishankar Krishnaswamy, Amit Kumar, Viswanath Nagarajan, Yogish Sabharwal,
  and Barna Saha.
\newblock The matroid median problem.
\newblock In {\em Proceedings of the Twenty-Second Annual {ACM-SIAM} Symposium
  on Discrete Algorithms}, pages 1117--1130, 2011.
\newblock \href {https://doi.org/10.1137/1.9781611973082.84}
  {\path{doi:10.1137/1.9781611973082.84}}.

\bibitem{krishnaswamy2018constant}
Ravishankar Krishnaswamy, Shi Li, and Sai Sandeep.
\newblock Constant approximation for $k$-median and $k$-means with outliers via
  iterative rounding.
\newblock In {\em Proceedings of the 50th Annual {ACM} {SIGACT} Symposium on
  Theory of Computing}, pages 646--659, 2018.
\newblock \href {https://doi.org/10.1145/3188745.3188882}
  {\path{doi:10.1145/3188745.3188882}}.

\bibitem{kumar2012constant}
Amit Kumar.
\newblock Constant factor approximation algorithm for the knapsack median
  problem.
\newblock In {\em Proceedings of the Twenty-Third Annual {ACM-SIAM} Symposium
  on Discrete Algorithms}, pages 824--832, 2012.
\newblock \href {https://doi.org/10.1137/1.9781611973099.66}
  {\path{doi:10.1137/1.9781611973099.66}}.

\bibitem{li2016approximating}
Shi Li and Ola Svensson.
\newblock Approximating $k$-median via pseudo-approximation.
\newblock {\em {SIAM} J. Comput.}, 45(2):530--547, 2016.
\newblock \href {https://doi.org/10.1137/130938645}
  {\path{doi:10.1137/130938645}}.

\bibitem{nemhauser1978analysis}
George~L. Nemhauser, Laurence~A. Wolsey, and Marshall~L. Fisher.
\newblock An analysis of approximations for maximizing submodular set functions
  - {I}.
\newblock {\em Math. Program.}, 14(1):265--294, 1978.
\newblock \href {https://doi.org/10.1007/BF01588971}
  {\path{doi:10.1007/BF01588971}}.

\bibitem{swamy2016improved}
Chaitanya Swamy.
\newblock Improved approximation algorithms for matroid and knapsack median
  problems and applications.
\newblock {\em {ACM} Trans. Algorithms}, 12(4):49:1--49:22, 2016.
\newblock \href {https://doi.org/10.1145/2963170} {\path{doi:10.1145/2963170}}.

\end{thebibliography}

\appendix

\section{Universal Matroid Median with Fixed Clients}\label{section:universal:fixed}

\begin{theorem}\label{theorem:matroid:universal:fixed}
There exists a bi-criteria $(O(1),O(1))$-approximation algorithm for the universal matroid median with fixed clients problem.
\end{theorem}

We closely follow the framework by Ganesh~\etalcite{ganesh2021universal} and give a proof sketch for~\cref{theorem:matroid:universal:fixed}. The formulation of universal matroid median with fixed clients is the same as universal matroid median, except that there exist fixed clients $C_f\subseteq\calC$ and each client realization must contain $C_f$.
In the sequel, we assume the existence of a deterministic $\gamma$-approximation for matroid median, for which we substitute the deterministic 8-approximation by Swamy~\cite{swamy2016improved} at the end. There then exists a $\gamma$-approximation for incremental matroid median, that is, given any subset of facilities in place, find the independent set containing all the given facilities which minimizes the median objective. The proof is the same as Theorem~24 in~\cite{ganesh2021universal} (the arXiv version), thus omitted here.
A similar LP relaxation is the following.
\begin{alignat}{3}
\text{min}\quad&&r&\geq0\tag{$\mathrm{M\text{-}LP}_{fix}$}\label{lp:universal:matroid:fixed}\\
\text{s.t.}\quad&&y(S)&\leq r_\calM(S)&&\forall S\subseteq\calF\notag\\
&& 0\leq x_{ij}&\leq y_i&&\forall i,j\notag\\
&&\sum_{i\in\calF}x_{ij}&=1&&\forall j\in\calC\notag\\
&&\sum_{j\in C_1}\sum_{i\in\calF}x_{ij}c_{ij}&\leq\optimal(C_1)+r\quad&&\forall C_1\subseteq\calC,\,C_1\supseteq C_f.\label{lp:universal:matroid:fixed:con4}
\end{alignat}

Likewise, in order to approximately solve \ref{lp:universal:matroid:fixed}, we need an approximate separation oracle on \eqref{lp:universal:matroid:fixed:con4}. By defining
\[
f_{x,y}(S')\defeq\sum_{j\in\calC\setminus C_f}\left(
-c_{jS'}+\sum_{i\in\calF}x_{ij}c_{ij}
\right)^+,
\]
separating \eqref{lp:universal:matroid:fixed:con4} is equivalent to determining whether there exists $S\in\calI$ s.t.
\begin{equation}
f_{x,y}(S)>\solution(C_f,S)-\sum_{j\in C_f}\sum_{i\in\calF}x_{ij}c_{ij}+r.\label{eqn:separation:knap}
\end{equation}
Though the LHS is still monotone submodular on $S\in\calI$, the RHS is changing with $S$, thus we need to consider \eqref{eqn:separation:knap} based on different values of $\solution(C_f,S)$. The basic idea is, for each possible value $M$, we only consider those $S\in\calI$ that satisfies $\solution(C_f,S)=M$, and maximize $f_{x,y}$ subject to $S\in\calI,\,\solution(C_f,S)=M$. If every such optimal $S_M$ satisfies $f_{x,y}(S_M)\leq M-\sum_{j\in C_f}\sum_{i\in\calF}x_{ij}c_{ij}+r$, $(x,y,r)$ is feasible to \ref{lp:universal:matroid:fixed}, otherwise it violates \eqref{lp:universal:matroid:fixed:con4} for some $C_1\subseteq\calC$ that attains $f_{x,y}(S_M)$ via $C_1\setminus C_f$, i.e., $C_1=C_f\cup\{j\in\calC\setminus C_f:\sum_{i\in\calF}x_{ij}c_{ij}\geq c_{jS_M}\}$.

To simplify the algorithm, we consider a reduced number of values $M\in\{\min\{(1+\epsilon)^s,\gamma(1+\epsilon)|C_f|\max_{i,j}c_{ij}\}:s\geq0\}\cup\{0\}$, and relax the constraint by trying to approximately maximize $f_{x,y}(S)$ subject to $\solution(C_f,S)\leq M,\,S\in\calI$. For a specific value of $M$ and the obtained solution $S_M\in\calI$, if $S_M$ satisfies $f_{x,y}(S_M)>M-\sum_{j\in C_f}\sum_{i\in\calF}x_{ij}c_{ij}+r$, it also satisfies \eqref{eqn:separation:knap} and thus violates \eqref{lp:universal:matroid:fixed:con4}, since $S_M$ by definition is $M$-cheap, i.e., it satisfies $\solution(C_f,S_M)\leq M$ (recall that we can efficiently compute $f_{x,y}(S_M)$ and $\solution(C_f,S_M)$). Otherwise, the algorithm assumes \eqref{eqn:separation:knap} is NOT satisfied by any $M$-cheap independent set, and moves on to check larger values of $M$.

Let $\calI_M$ be the set of \emph{bases} of $\calI$ satisfying $\solution(C_f,S)\leq M$. It follows that the maximum of $f_{x,y}$ over $M$-cheap independent sets must be attained in $\calI_M$, since $f_{x,y}$ is monotone submodular and $\solution(C_f,S)$ is non-increasing. Then maximizing $f_{x,y}$ is equivalent to maximizing $f_{x,y}$ on a 1-system, where the maximal independent sets are exactly $\calI_M$. We use the same greedy oracle in~\cite{ganesh2021universal}, where we replace their incremental $k$-median algorithm with the aforementioned $\gamma$-approximate algorithm for incremental matroid median. The same approximate separation oracle then outputs a $(2\gamma(1+\epsilon),2)$-approximate fractional solution $(\bar x,\bar y,\bar r)$ to universal matroid median with fixed clients in polynomial time, i.e.,\[\sum_{j\in C_1}\sum_{i\in\calF}\bar x_{ij}c_{ij}\leq2\gamma(1+\epsilon)\optimal(C_1)+2\minregret,\,\forall C_1\subseteq\calC,\,C_1\supseteq C_f.\]

We let $\fav(j)=\sum_{i\in\calF}\bar x_{ij}c_{ij}$ be the cost of client $j$ in the fractional solution.
We proceed to use our bi-criteria $(\atau',\btau')$-approximation for \matmd with discount values $r_j=8\fav(j)$ if $j\notin C_f$ and $r_j=0$ if $j\in C_f$, and we obtain
\begin{align*}
\minregret&=\max_{C_1\subseteq\calC,C_1\supseteq C_f}\left\{\sum_{j\in C_1}m_j-\optimal(C_1)\right\}\geq\max_{C_1\subseteq\calC,C_1\supseteq C_f}\left\{\sum_{j\in C_1}m_j-8\fav(j)\right\}\\
&=\sum_{j\in C_f}(m_j-8\fav(j))+\sum_{j\in\calC\setminus C_f}(m_j-8\fav(j))^+,
\end{align*}
using the small integrality bound of matroid median~\cite{swamy2016improved}, thus yielding
\[
\minregret+8\sum_{j\in C_f}\fav(j)\geq\sum_{j\in C_f}m_j+\sum_{j\in\calC\setminus C_f}(m_j-8\fav(j))^+=\sum_{j\in\calC}(m_j-r_j)^+.
\]

Our result then reveals a solution $\baring{S}\in\calI$ s.t.
\begin{align*}
&\sum_{j\in\calC}(c_{j\baring{S}}-\atau'r_j)^+\leq\btau'\sum_{j\in\calC}(m_j-r_j)^+\leq\btau'\left(\minregret+8\sum_{j\in C_f}\fav(j)\right)\\
\Rightarrow&\sum_{j\in C_f}(c_{j\baring{S}}-8\btau'\fav(j))+\sum_{j\in\calC\setminus C_f}(c_{j\baring{S}}-8\atau'\fav(j))^+\leq\btau'\minregret,
\end{align*}
therefore for each $C_1\subseteq\calC,\,C_1\supseteq C_f$, one has
\begin{align*}
\sum_{j\in C_1}c_{j\baring{S}}&\leq\sum_{j\in C_f}(c_{j\baring{S}}-8\btau'\fav(j))+\sum_{j\in C_1\setminus C_f}(c_{j\baring{S}}-8\atau'\fav(j))^++8\max\{\atau',\btau'\}\sum_{j\in C_1}\fav(j)\\
&\leq\btau'\minregret+8\max\{\atau',\btau'\}\left(2\gamma(1+\epsilon)\optimal(C_1)+2\minregret
\right),
\end{align*}
whence a bi-criteria $(O(1),O(1))$-approximation algorithm for universal matroid median with fixed clients follows.

\end{document}